\documentclass[11pt]{article}
\usepackage{fullpage}
\usepackage{graphicx}
\usepackage{amsmath}
\usepackage{amssymb}
\usepackage{amsthm}
\usepackage{enumerate}
\usepackage[pdfpagelabels, pagebackref, naturalnames]{hyperref}
\usepackage[font=small, format=hang, labelfont=bf]{caption}
\usepackage[labelfont=default]{subcaption}
\usepackage{xspace}
\usepackage{todonotes}

\newcommand{\Plane}{\ensuremath{\mathbb{R}^2}}

\newcommand{\bd}{\ensuremath{\partial}}

\newcommand{\cut}{\ensuremath{\mathrm{cut}}}
\newcommand{\poc}{\ensuremath{\mathrm{poc}}}

\newcommand{\Kernel}{\mathcal{K}}

\newcommand{\RKernel}{\mathcal{R}}

\newtheorem{lemma}{Lemma}
\newtheorem{theorem}{Theorem}

\newtheorem{observation}{Observation}

\theoremstyle{definition}

\graphicspath{{./}}

\newcommand\Tstrut{\rule{0pt}{2.6ex}}       
\newcommand\Bstrut{\rule[-1.1ex]{0pt}{0pt}} 

\let\geq\geqslant
\let\leq\leqslant

\title{Improved Bounds for Beacon-Based Coverage and Routing in Simple Rectilinear Polygons
\thanks{%
Work by S.W. Bae was supported by Basic Science Research Program through the National
Research Foundation of Korea (NRF) funded by the Ministry of Science, ICT \& Future Planning
(2013R1A1A1A05006927).
Work by C.-S. Shin was supported by Research Grant of Hankuk University of Foreign Studies.
}
}

\author{%
Sang Won Bae\footnote{%
Department of Computer Science, Kyonggi University, Suwon, Korea.
Email: \texttt{swbae@kgu.ac.kr}
}
\and %
Chan-Su Shin\footnote{%
Division of Electronic Systems \& Computer Engineering,
Hankuk University of Foreign
Studies, Yongin, Korea.
Email: \texttt{cssin@hufs.ac.kr, chansu@gmail.com}.
}
\and %
Antoine Vigneron\footnote{%
Visual Computing Center,
King Abdullah University of Science and	Technology (KAUST), Thuwal 23955-6900, Saudi Arabia.
Email: \texttt{antoine.vigneron@kaust.edu.sa}
}
}

\date{%
\today
}

\begin{document}
\maketitle

\begin{abstract}
We establish tight bounds for beacon-based coverage problems, and improve the bounds for beacon-based routing problems in simple rectilinear polygons. 
Specifically, we show that $\lfloor \frac{n}{6} \rfloor$ beacons are always sufficient and sometimes necessary to cover a simple rectilinear polygon $P$ with $n$ vertices. 
 We also prove tight bounds for the case where $P$ is monotone,
and we present an optimal linear-time algorithm that computes the beacon-based kernel
of $P$.
For the routing problem, we show that $\lfloor \frac{3n-4}{8} \rfloor - 1$ beacons
are always sufficient, and $\lceil \frac{n}{4}\rceil-1$ beacons are sometimes necessary to route between all pairs of points in $P$.
\end{abstract}

\section{Introduction} \label{sec:intro}

A \emph{beacon} is a facility or a device that attracts objects within a given domain.
We assume that objects in the domain, such as mobile agents or robots, know the exact location
or the direction towards an activated beacon in the domain, even if it is not directly visible.
More precisely, given a polygonal domain $P$, a beacon is placed at a fixed point in $P$.
When a beacon $b \in P$ is activated, an object $p \in P$ moves along the ray
starting at  $p$ and towards the beacon $b$ until it either  hits the boundary $\bd P$ of $P$,
or it reaches $b$. (See Figure~\ref{fig:intro}a.)
If $p$ hits an edge $e$ of $P$, then it continues to move along $e$ in the direction such
that the Euclidean distance to $b$ decreases.
When $p$ reaches an endpoint of $e$, it may move along the ray from the current
position of $p$ towards $b$, if possible, until it again hits the boundary $\bd P$ of $P$.
So, $p$ is pulled by $b$ in a greedy way, so that the Euclidean distance to $b$
is monotonically decreasing, as an iron particle is pulled by a magnet.
There are two possible outcomes: Either $p$ finally reaches $b$, or it stops at a local minimum,
called a \emph{dead point}, where there is no direction in which, locally, the distance to $b$ 
strictly decreases. In the former case, $p$ is said to be \emph{attracted} by the beacon $b$.

This model of beacon attraction was recently suggested by 
Biro~\cite{b-bbrg-13, bgikm-cccg-13, bikm-wads-13},
and extends the classical notion of visibility.
We consider two problems based on this model:
the \emph{coverage} and the \emph{routing} problem,
introduced by Biro~\cite{b-bbrg-13} and Biro et al.~\cite{bgikm-cccg-13,bikm-wads-13}.
In the beacon-based coverage problem, we need to place beacons in $P$
so that any point $p\in P$ is attracted by at least one of the beacons.
In this case, we say the set of beacons \emph{covers} or \emph{guards} $P$.
In the beacon-based routing problem, we want to place beacons in $P$
so that every pair $s, t\in P$ of points can be routed:
We say that $s$ is  \emph{routed} to $t$ if there is a sequence of beacons in $P$
that can be activated and deactivated one at a time,
such that the source $s$ is successively attracted by each of beacon of this sequence,
and finally reaches the target $t$, which is regarded as a beacon.

In this paper, we are interested in combinatorial bounds
on the number of beacons required for coverage and routing,
in particular when the given domain $P$ is a simple rectilinear polygon.
Our bounds are variations on visibility-based guarding results, such
as the well-known \emph{art gallery theorem}~\cite{c-ctpg-75} and
its relatives~\cite{o-apragt-83,ghks-ggprp-96, kkk-tgrfw-83, g-spragt-86, mp-agtgg-03}.
The beacon-based coverage problem is analogous to the art gallery problem,
while the beacon-based routing problem is analogous to the \emph{guarded guards} 
problem~\cite{mp-agtgg-03}, 
which asks for a set of point guards in $P$ such that every point is visible from
at least one guard and every guard is visible from another guard.
For the art gallery problem, it is known that
$\lfloor \frac{n}{3} \rfloor$ point guards are sufficient, and sometimes necessary,
to guard a simple polygon $P$ with $n$ vertices~\cite{c-ctpg-75}.
If $P$ is rectilinear, then $\lfloor \frac{n}{4} \rfloor$ are necessary and 
sufficient~\cite{kkk-tgrfw-83,o-agta-87, g-spragt-86}.
In the guarded guards problem, this number
becomes $\lfloor \frac{3n-1}{7} \rfloor$ for simple polygons
and $\lfloor \frac{n}{3} \rfloor$ for simple rectilinear polygons~\cite{mp-agtgg-03}.
Other related results are mentioned in the book~\cite{o-agta-87} by O'Rourke
or the surveys by Shermer~\cite{s-rrag-92} and Urrutia~\cite{u-agip-00}.

Biro et al.~\cite{bgikm-cccg-13} initiated research on combinatorial bounds
for beacon-based coverage and routing problems,
with several nontrivial bounds for different types of domains such as
rectilinear or non rectilinear polygons, with or without holes.
When the domain $P$ is a simple rectilinear polygon with $n$ vertices, they showed that $\lfloor \frac{n}{4} \rfloor$ beacons are sufficient to cover
any rectilinear polygon with $n$ vertices,
while $\lfloor \frac{n+4}{8} \rfloor$ beacons are necessary to cover
the same example in \figurename~\ref{fig:intro},
and conjectured that $\lfloor \frac{n+4}{8} \rfloor$ would be the tight bound.
They also proved that $\lfloor \frac{n}{2} \rfloor -1$ beacons are always sufficient for routing, and some domains, such as the domain depicted in \figurename~\ref{fig:intro}a, require $\lfloor \frac{n}{4} \rfloor - 1$ beacons.

\begin{table}[tbh]
\centering
\begin{tabular}{c|c|c|c|c}
 \hline
 & \multicolumn{2}{c|}{Lower bound} & \multicolumn{2}{c}{Upper bound}  \Tstrut\Bstrut\\
 \cline{2-5}
 & Best results & Our results & Best results & Our results
 \Tstrut\Bstrut\\ 
  \hline\hline
 Coverage & $\lfloor \frac{n+4}{8} \rfloor$ \quad\; \cite{bgikm-cccg-13}&
 $\lfloor \frac{n}{6}\rfloor$ \qquad\, [Theorem~\ref{thm:covering}] & $\lfloor \frac{n}{4}\rfloor$ \quad\enspace\;\; \cite{bgikm-cccg-13} & $\lfloor \frac{n}{6}\rfloor$ \qquad\quad\enspace [Theorem~\ref{thm:covering}]
  \Tstrut\Bstrut\\ \hline
 Routing & $\lfloor \frac{n}{4}\rfloor-1$  \enspace \cite{bgikm-cccg-13} & $\lceil \frac{n}{4}\rceil -1$ \enspace [Theorem~\ref{thm:routing}] & $\lfloor \frac{n}{2}\rfloor -1$  \enspace \cite{bgikm-cccg-13} &  $\lfloor \frac{3n-4}{8}\rfloor-1$ \enspace [Theorem~\ref{thm:routing}]
  \Tstrut\Bstrut\\
 \hline
\end{tabular}
\caption{Best known results and our results on the number of beacons required for beacon-based coverage and routing in a simple rectilinear polygon $P$ with $n$ vertices. Our lower bound on the routing problem holds for any $n \neq 6$.}
\label{tbl:summary}
\end{table}

\paragraph{Our results.}

In this paper, we first prove tight bounds on beacon-based coverage problems for simple rectilinear polygons. (See \tablename~\ref{tbl:summary}.) 
In Section~\ref{sec:coverage},
we give a lower bound construction that requires $\lfloor \frac{n}{6} \rfloor$ beacons, and then we present a method of placing the same number of beacons to cover $P$, which matches the lower bound we have constructed.
These results settle the open questions on the beacon-based coverage problems for simple rectilinear polygons
posed by Biro et al.~\cite{bgikm-cccg-13}.
We also consider the case of monotone polygons:
For routing in a monotone rectilinear polygon, the same bound 
$\lfloor \frac{n}{4} \rfloor - 1$ holds,
while $\lfloor \frac{n+4}{8} \rfloor$ beacons are always sufficient to cover a monotone rectilinear polygon.

We next improve the upper bound from $\lfloor \frac{n}{2}\rfloor-1$ to $\lfloor \frac{3n-4}{8} \rfloor - 1$ for
the routing problem in Section~\ref{sec:routing}. We also slightly improve the lower bound from $\lfloor \frac{n}{4}\rfloor - 1$ to $\lceil \frac{n}{4}\rceil -1$ for any $n\neq 6$. 

We also present an optimal linear-time algorithm that computes the \emph{beacon kernel}
$\Kernel(P)$ of a simple rectilinear polygon $P$ in Section~\ref{sec:kernel}.
The {beacon kernel} $\Kernel(P)$ of $P$ is defined to be
the set of points $p\in P$ such that placing a single beacon at $p$ is sufficient to completely cover $P$.
Biro first presented an $O(n^2)$-time algorithm that computes the kernel $\Kernel(P)$
of a simple polygon $P$ in his thesis~\cite{b-bbrg-13},
and  Kouhestani et al.~\cite{krs-cccg-14} soon improved it to $O(n\log n)$ time
with the observation that $\Kernel(P)$ has a linear complexity.
Our algorithm is based on a new, yet simple, characterization of the kernel $\Kernel(P)$.

\begin{figure}[tb]
\centering
\includegraphics[width=\textwidth]{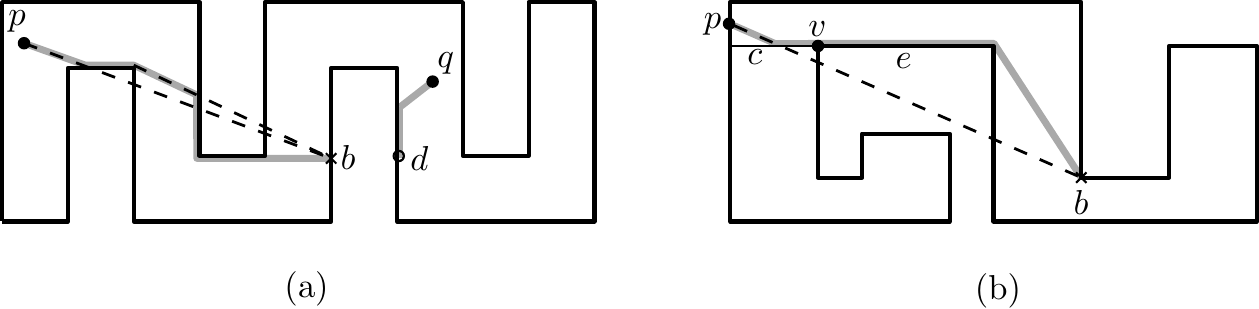}
\caption{(a) A lower bound construction $P$ by Biro et al.~\cite{bgikm-cccg-13}.
 A point $p \in P$ is attracted by a beacon $b$ through the beacon attraction path
 depicted by the thick gray path, while $q \in P$ is not since it stops at the dead point $d$.
 (b) Another rectilinear polygon $P$. If one partitions $P$ by the horizontal cut $c$ (dashed segment) 
 at $v$ into two subpolygons $P^+_c$ and $P^-_c$ and handle each separately,
 then it does not guarantee that $P$ is guarded.
 In this case, $p\in P^+_c$ is attracted by $b$ inside the subpolygon $P^+_c$
 while it is not the case in the whole domain $P$.}
\label{fig:intro}
\end{figure}

\section{Preliminaries} \label{sec:pre}

A \emph{simple rectilinear polygon} is a simple polygon whose edges are either horizontal or vertical.
The internal angle at each vertex of a rectilinear polygon is always $90^\circ$ or $270^\circ$.
We call a vertex with internal angle $90^\circ$ a \emph{convex vertex}, and a vertex with internal 
angle $270^\circ$ is called a \emph{reflex vertex}.
For any simple rectilinear polygon $P$, we let $r = r(P)$ be the number of its reflex vertices.
If $P$ has $n$ vertices in total, then $n = 2r + 4$, because the sum of the signed turning angles
along $\bd P$ is $360^\circ$.
An edge of $P$ between two  convex vertices is called a \emph{convex edge},
and an edge between two reflex vertices is called a \emph{reflex edge}.
Each convex or reflex edge $e$ shall be called \emph{top}, \emph{bottom}, \emph{left} or \emph{right}
according to its orientation:
If $e$ is horizontal and the two adjacent edges of $e$ are downwards from $e$,
then $e$ is a top convex or reflex edge.
(The edge $e$ in \figurename~\ref{fig:intro}b is a top reflex edge.)
For each edge $e$ of $P$, we are often interested in the half-plane $H_e$
whose boundary supports $e$ and whose interior includes the interior of $P$ locally at $e$.
We shall call $H_e$ the \emph{half-plane supporting $e$}.

A rectilinear polygon $P$ is called \emph{$x$-monotone} (or \emph{$y$-monotone})
if any vertical (resp., horizontal) line intersects $P$ in at most one connected component.
If $P$ is both $x$-monotone and $y$-monotone, then $P$ is said to be \emph{$xy$-monotone}.
From the definition of the monotonicity,
we observe the following.
\begin{observation} \label{obs:monotone}
A rectilinear polygon $P$ is $x$-monotone if and only if $P$ has no vertical reflex edge.
Hence, $P$ is $xy$-monotone if and only if $P$ has no reflex edge.
\end{observation}

Our approach to attain  tight upper bounds relies on
partitioning a given rectilinear polygon $P$ into subpolygons by cuts.
More precisely, a \emph{cut} in $P$ is a chord\footnote{%
A \emph{chord} $c$ of a polygon is a line segment between two points on the boundary
such that all points on $c$ except the two endpoints lie in the interior of the polygon.}
of $P$ that is horizontal or vertical.
There is a unique cut at a point $p$ on the boundary $\bd P$ of $P$
unless $p$ is a vertex of $P$.
If $p$ is a reflex vertex, then there are two cuts at $p$,
one of which is \emph{horizontal} and the other is \emph{vertical},
while there is no cut at $p$ if $p$ is a convex vertex.
Any horizontal cut $c$ in $P$ partitions $P$ into two subpolygons:
one below $c$, denoted by $P_c^-$, and the other above $c$ denoted by $P_c^+$.
Analogously, for any vertical cut $c$,
let $P_c^-$ and $P_c^+$ denote the subpolygons to the left and to the right of $c$, respectively.

For a beacon $b$ and a point $p\in P$, the \emph{beacon attraction path} of $p$ with respect to $b$,
or simply the \emph{$b$-attraction path} of $p$,
is the piecewise linear path from $p$ created by the attraction of $b$
as described in Section~\ref{sec:intro}.
(See \figurename~\ref{fig:intro}a.)
If the $b$-attraction path of $p$ reaches $b$, then we say that $p$ is \emph{attracted} to $b$.
As was done for the classical visibility notion~\cite{o-apragt-83,ghks-ggprp-96},
a natural approach would find a partition of $P$ into smaller subpolygons of similar size,
and handle them recursively.
However, we must be careful when choosing a partition of $P$, because
an attraction path within a subpolygon may not be an attraction path
within $P$. (See \figurename~\ref{fig:intro}b.)
So $P$ is not necessarily guarded by the union of the guarding sets of the subpolygons.

Thus, when applying a cut in $P$, we want to make sure that
beacon attraction paths in a subpolygon $Q$ of $P$ do not \emph{hit} the new edge
of $Q$ produced by $c$.
To be more precise, we say that an edge $e$ of $P$ is \emph{hit} by $p$ with respect to $b$
if the $b$-attraction path of $p$ makes a bend along $e$.
\begin{observation} \label{obs:convex_edge}
 Let $b$ be a beacon in $P$ and $p\in P$ be any point
 such that $p$ is attracted by $b$.
 If the $b$-attraction path of $p$ hits an edge $e$ of $P$,
 then $p\in H_e$ and $b\notin H_e$, where $H_e$ denotes the half-plane supporting $e$.
 Therefore, no beacon attraction path hits a convex edge of $P$.
\end{observation}
Thus, if we choose a cut that becomes a convex edge on both sides,
then we will be able to handle each subpolygon separately.

In this paper, we make the  general position assumption 
that \textit{no cut in $P$ connects two reflex vertices.}
This general position can be obtained by perturbing the reflex vertices of $P$ locally,
and such a perturbation does not harm the upper bounds on our problems in general.
It will be discussed in the full version of the paper.

\section{The Beacon Kernel} \label{sec:kernel}

Before continuing to the beacon-based coverage problem,
we consider simple rectilinear polygons that can be covered by a single beacon.
This is related to the  \emph{beacon kernel} $\Kernel(P)$ of a simple polygon $P$,
defined to be the set of all points $p \in P$
such that a beacon placed at $p$ attracts all points in $P$.
Specifically, we give a characterization of rectilinear polygons $P$ such that $\Kernel(P) \neq \emptyset$.
Our characterization is  simple and constructive, resulting in a linear-time algorithm
that computes the beacon kernel $\Kernel(P)$ of any simple rectilinear polygon $P$.

Let $R$ be the set of reflex vertices of $P$. Let $v\in R$ be any reflex vertex with
two incident edges $e_1$ and $e_2$. 
For $i \in \{1, 2\}$,
define $N_i$ to be the closed half-plane whose boundary is the line orthogonal to $e_i$ through $v$
and whose interior excludes $e_i$.
Let $C_v := N_1 \cup N_2$.
Observe that $C_v$ is a closed cone with apex $v$.
Biro~\cite[Theorem 5.2.8]{b-bbrg-13} showed that the kernel $\Kernel(P)$ of $P$
is the set of points in $P$ that lie in $C_v$ for all reflex vertices $v\in R$:
\begin{lemma}[Biro~\cite{b-bbrg-13}] \label{lem:kernel_biro}
 For any simple polygon $P$ with set $R$ of reflex vertices, it holds that
 \[ \Kernel(P) = \left(\bigcap_{v\in R} C_v\right) \cap P = P\setminus \left(\bigcup_{v\in R}\overline{C_v} \right),\]
 where $\overline{C_v} = \Plane \setminus C_v$ denotes the complement of $C_v$. 
\end{lemma}

Note that Lemma~\ref{lem:kernel_biro} holds for any simple polygon $P$.
We now assume that $P$ is a simple rectilinear polygon.
Then, for any reflex vertex $r\in R$, the set $C_v$ forms a closed cone with aperture angle $270^\circ$
whose boundary consists of two rays following the two edges incident to $v$.
Let $R_1 \subseteq R$ be the set of reflex vertices incident to a reflex edge,
and let $R_2 := R \setminus R_1$.
So a vertex in $R_1$ is adjacent to at least one reflex vertex that also belongs to $R_1$,
and a vertex in $R_2$ is always adjacent to two convex vertices.
We then observe the following.
\begin{lemma}
\label{lem:kernel_lem}
For any simple rectilinear polygon $P$,
\[ \left( \bigcap_{v\in R_1}C_v \right) \cap P  \subseteq  \left( \bigcap_{v\in R_2}C_v \right) \cap P.\]
\end{lemma}
\begin{proof}
For a contradiction, suppose that there exists a point $p \in P$ that is included in
$\bigcap_{v\in R_1}C_v$ but avoids $\bigcap_{v\in R_2}C_v$.
Then, there must exist a reflex vertex $v\in R_2$ such that $p \notin C_v$,
or equivalently, $p \in \overline{C_v}$.
That is, $\bigcap_{v'\in R_1}C_{v'}$ and $\overline{C_v}$ have a nonempty intersection.
Let $u$ and $w$ be the two vertices adjacent to $v$ such that $u$, $v$, and $w$ appear
on $\bd P$ in counterclockwise order.
Note that both $u$ and $w$ are convex since $v\in R_2$.

\begin{figure}[tb]
\centering
\includegraphics[width=0.7\textwidth]{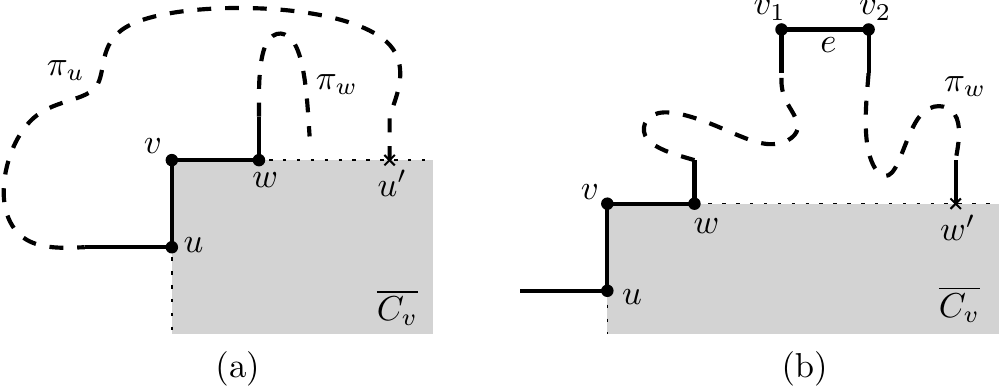}
\caption{Proof of Lemma~\ref{lem:kernel_lem}.}
\label{fig:kernel}
\end{figure}

Since $\overline{C_v} \cap P \neq \emptyset$ and $\overline{C_v}$ is an open set,
the boundary $\bd P$ of $P$ crosses $\bd \overline{C_v}$ at some points other than the two edges $uv$ and $vw$.
Let $w'$ be the first point in $\bd P \cap \bd \overline{C_v}$
that we encounter when traveling along $\bd P$ counterclockwise, starting at $w$.
Analogously, let $u'$ be the first point in $\bd P \cap \bd \overline{C_v}$
that we encounter when traveling along $\bd P$ clockwise, starting at $u$.
Let $\pi_w\subset \bd P$ and $\pi_u\subset \bd P$ be the paths described above
from $w$ to $w'$ and from $u$ to $u'$, respectively.
As  $\pi_w$ and $\pi_u$ are subpaths of $\bd P$, they do not intersect,
and we have $w' \notin uv$ and $u' \notin vw$.

The boundary $\bd \overline{C_v}$ of $\overline{C_v}$ consists of two rays $\rho_w$ and $\rho_u$,
starting from $v$ towards $w$ and $u$, respectively.
We claim that either $w'$ lies on $\rho_w$ or $u'$ lies on $\rho_u$.
(See \figurename~\ref{fig:kernel}a.) 
Indeed, suppose that $u' \notin \rho_u$.
Then $\pi_w$ should be contained in the region bounded by the simple closed curve
$\pi_u \cup u'v \cup vu$, since $\pi_w$ does not intersect $\overline{C_v} \cup uv$.
This implies that $w'$ must lie on $wu' \subset \rho_w$.
Hence, our claim is true.

Without any loss of generality, we assume that $w'\in \rho_w$,
the edge $vw$ is horizontal, and the interior of $P$ lies locally above $vw$,
as shown in \figurename~\ref{fig:kernel}b.
Then, the path $\pi_w$ must contain at least one top reflex edge $e$ lying above
the line through $w$ and $w'$,
since $w$ and $w'$ have the same $y$-coordinate and $\pi_w$ avoids $\overline{C_v}$.
Let $v_1$ and $v_2$ be the two endpoints of $e$,
so $v_1, v_2 \in R_1$.
Then $C_{v_1} \cap C_{v_2}$ is the half-plane $H_e$ supporting $e$.
Since $e$ is a top reflex edge, $H_e \cap \overline{C_v} = \emptyset$.
This is a contradiction to our assumption that $\bigcap_{v'\in R_1} C_{v'}$ intersects $\overline{C_v}$.
\end{proof}

Let $\RKernel(P)$ be the intersection of the half-planes $H_e$ supporting $e$
over all reflex edges $e$ of $P$. We conclude the following.
\begin{theorem} \label{thm:kernel}
 Let $P$ be a simple rectilinear polygon.
 A point $p \in P$ lies in its beacon kernel $\Kernel(P)$ if and only if
 $p \in H_e$ for any reflex edge $e$ of $P$.
 Therefore, it always holds that $\Kernel(P) = \RKernel(P) \cap P$, and
 the kernel $\Kernel(P)$ can be computed in linear time.
\end{theorem}
\begin{proof}
Recall that $C_v$ for any $v\in R_1$ forms a cone with apex $v$ and aperture angle $270^\circ$.
Since any $v\in R_1$ is adjacent to another reflex vertex $w\in R_1$
the intersection $C_v \cap C_w$ forms exactly the half-plane $H_e$ supporting the reflex edge $e$
with endpoints $v$ and $w$.
It implies that $\RKernel(P) = \bigcap_{v\in R_1} C_v$.
So by Lemma~\ref{lem:kernel_lem}, we have
\[ \Kernel(P) = \bigcap_{v \in R} C_v \cap P = \bigcap_{v \in R_1} C_v \cap P = \RKernel(P) \cap P.\]

The set $\RKernel(P)$ is an intersection of axis-parallel halfplanes,
so it is a (possibly unbounded) axis-parallel rectangle.
In order to compute the kernel $\Kernel(P)$,
we identify the extreme reflex edge in each of the four directions to compute $\RKernel(P)$,
and then intersect it with $P$.
This can be done in linear time.
\end{proof}

\section{Beacon-Based Coverage} \label{sec:coverage}

In this section, we study the beacon-based coverage problem for rectilinear polygons.
A set of beacons in $P$ is said to \emph{cover} or \emph{guard} $P$
if and only if every point $p\in P$ can be attracted by at least one of them.

Our main result is the following. 
\begin{theorem} \label{thm:covering}
 Let $P$ be a simple rectilinear polygon $P$ with $n \geq 6$ vertices and
 $r \geq 1$ reflex vertices.
 Then $\lfloor \frac{n}{6}\rfloor = \lceil \frac{r}{3} \rceil$ beacons are sufficient
 to guard $P$, and sometimes necessary.
 Moreover, all these beacons can be placed at reflex vertices of $P$.
\end{theorem}

We now sketch the proof of Theorem~\ref{thm:covering}. The lower bound construction is 
a rectangular spiral $P_r$ consisting of a sequence of $r+1$ thin rectangles, as depicted in 
\figurename~\ref{fig:spirals}. The sequence of vertices $v_0v_1v_2 \dots v_r$, where $v_1\dots v_r$
are the reflex vertices of $P_r$,  form
a polyline called the {\it spine} of the spiral. The key idea is the following. Consider
the case $r=7$ (\figurename~\ref{fig:spirals}a). At first glance, it looks like the spiral
can be covered by two beacons $b_1$ and $b_2$ placed near $v_2$ and $v_6$, respectively.
However, at closer look, it appears that the small shaded triangular region on the bottom left 
corner is not covered. Hence, $P_7$ requires $3=\lceil \frac 7 3 \rceil$ beacons, as announced.
More generally, we can prove that  for a suitable choice of the edge lengths of the
spine of $P_r$, an optimal coverings for $P_r$ consists in placing a
beacon at every third rectangle of $P_r$, which yields the bound $\lceil \frac r 3 \rceil$.
The spine of $P_r$ is depicted in \figurename~\ref{fig:spirals}b and c, where the aspect
ratio of the rectangles is roughly $4+r/2$.

\begin{figure}[tb]
\centering
\includegraphics[width=\textwidth]{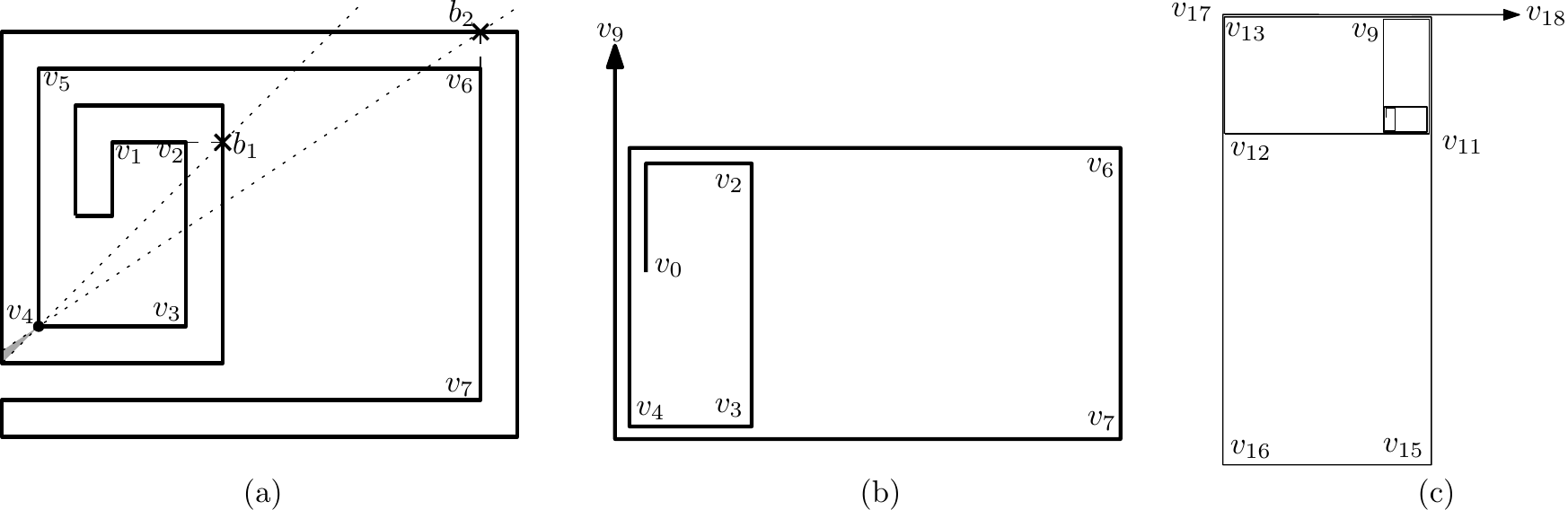}
\caption{Lower bound construction:
(a) Placing two beacons $b_1$ and $b_2$ in $P_7$ near $v_2$ and $v_6$
 is not enough to cover the shaded region near the reflex vertex $v_4$.
 (b)(c) The spine of our construction $P_r$ for $r=9$ and for $r=18$. }
\label{fig:spirals}
\end{figure}

The construction for the upper bound in Theorem~\ref{thm:covering} is more involved. 
We first prove that for any polygon with at most 3 reflex vertices, one beacon placed
at a suitable reflex vertex is sufficient. 
For a larger number $r \geq 4$ of reflex vertices, we proceed by induction.
So we partition $P$ using a cut, and we handle each side recursively.
As mentioned in Section~\ref{sec:pre}, the difficulty is that in some cases, the
union of the two guarding sets of the subpolygons do not cover $P$.
So we will first try to perform a {\it safe cut} $c$, that is, a cut $c$ which 
is not incident to any reflex vertex,  such that there is at least one reflex
vertex on each side, and such that 
$\lceil r(P_c^-)/3 \rceil + \lceil r(P_c^+)/3 \rceil =\lceil r/3 \rceil$.
(See \figurename~\ref{fig:covering_intro}a.)
If such a cut exists, then we can recurse on both side. By Observation~\ref{obs:convex_edge},
the union of the guarding sets of the two subpolygons guards $P$.
Unfortunately, some polygons do not admit any safe cut. In this case, we show
by a careful case analysis that we can always find a suitable cut. 
(See the example in \figurename~\ref{fig:covering_intro}b.)

\begin{figure}[tb]
\centering
\includegraphics[width=.7\textwidth]{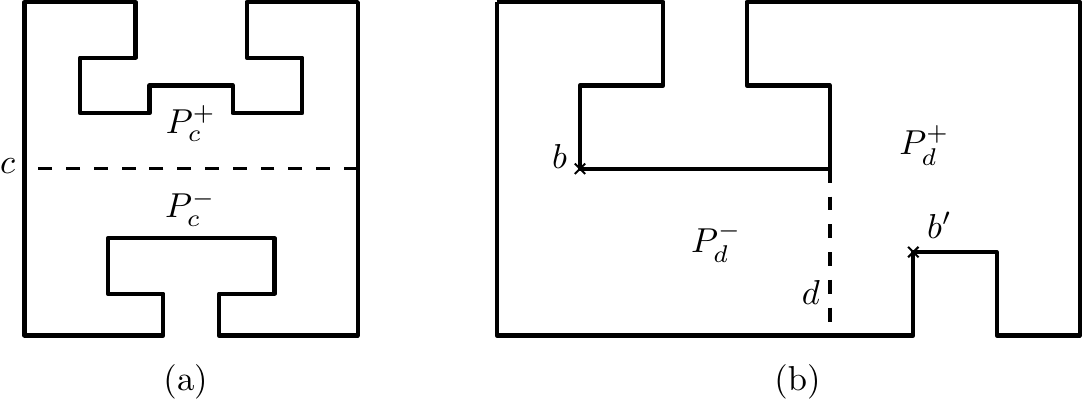}
\caption{Upper bound construction. 
(a) A safe cut $c$ of a polygon with $r=10$ reflex vertices.
(b) This polygon admits no safe cut. We cut along  $d$,
and the polygon is guarded by any two beacons $b$ and $b'$ placed at
reflex vertices of $P_d^-$ and $P_d^+$, respectively.
\label{fig:covering_intro}}
\end{figure}

\subsection{Proof of the lower bound for coverage}\label{subsec:lower_bound_coverage}

In this section, we prove the lower bound in Theorem~\ref{thm:covering}. Our construction is a 
spiral-like rectilinear polygon $P_r$ that cannot be guarded by less than  $\lceil \frac{r}{3} \rceil$
beacons. (See \figurename~\ref{fig:spirals}.) More precisely, a rectilinear polygon is called a \emph{spiral} if all its reflex vertices are consecutive along its boundary.

The \emph{spine} of a spiral $P$ with $r$ reflex vertices is
the portion of its boundary $\bd P$ connecting $r+2$ consecutive vertices $v_0, v_1, \ldots, v_{r+1}$
such that $v_1, \ldots, v_r$ are the reflex vertices of $P$. (See \figurename~\ref{fig:spirals3}a.)
Note that the two end vertices $v_0$ and $v_{r+1}$ of the spine of a spiral
are the only convex vertices that are adjacent to a reflex vertex.
The spine can also be specified by the sequence of edge lengths $(a_0, \ldots, a_r)$
such that $a_i$ is the length of the edge $v_iv_{i+1}$ for $i=0,\ldots, r$.

\begin{figure}[tb]
\centering
\includegraphics[width=\textwidth]{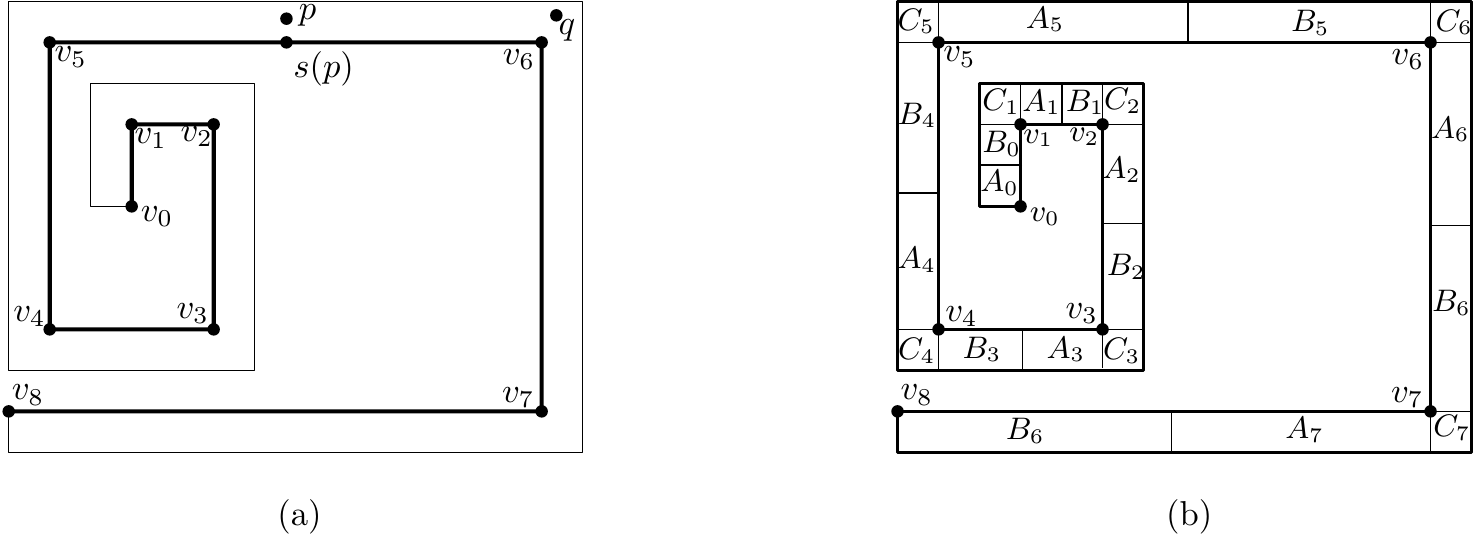}
\caption{(a) The spine (bold) of a spiral. The point $s(p)$ appears before
$v_6=s(q)$ along the spine, so $p \prec q$. (b) The partition of $P_7$ into rectangles $A_i,B_i,C_i$.
\label{fig:spirals3}}
\end{figure}

We define an order $\prec$ among points in any spiral $P$ as follows. Let $p,q$ be two points in $P$.
Let $s(p)$ and $s(q)$ denote the closest point to $p$ and $q$ on the spine, according to the geodesic
distance within $P$.  (See \figurename~\ref{fig:spirals3}a.) Then we say that $p$ precedes $q$, 
which we denote by $p \prec q$, if $s(p)$ precedes $s(q)$ along the spine, that is, $s(p)$ 
is on the portion of the spine between $v_0$ and $s(q)$.

We will use the following partition of a spiral $P$ with $r$ reflex vertices 
into $3r+2$ rectangular subpolygons. It is obtained
by applying the vertical and horizontal cuts at $v_i$ for each $i=1, \ldots, r$
and the cut at the midpoint of edge $v_iv_{i+1}$ for each $i=0, \ldots, r$.
We call these rectangles $A_0, B_0, C_1, A_1, B_1, \ldots, C_r, A_r, B_r$, ordered along
the spine. (See \figurename~\ref{fig:spirals3}b.)

For any integer $r\geq 0$, let $P_r$ be the spiral with $r$ reflex vertices
whose spine is determined by the following edge length sequence $(a_0, \ldots, a_r)$:
for any nonnegative integer $j$,
\[
  a_{2j} =  \begin{cases}
                    1 & (j = 0)\\
                    \rho^{2 \lfloor \frac{j}{3} \rfloor + 1} + j\epsilon & (j \geq 1)

                  \end{cases}
  \qquad\qquad
  a_{2j+1} =  \begin{cases}
                    1 + j\epsilon & (j \leq 1)\\
                    \rho^{2 \lfloor \frac{j+1}{3} \rfloor } + j\epsilon & (j \geq 2)
                  \end{cases},
\]
where $\epsilon > 0$ is a sufficiently small positive number
and $\rho > 2 + (r / 2 + 2)\epsilon$ is a constant.
(See \figurename~\ref{fig:spirals}.)

Therefore, the rectangles $A_i,B_i,C_i$ corresponding to $P_r$ are as follows, for any $i$.
Rectangle  $A_i$ and $B_i$ have side lengths $a_i/2$ and $w_i < \epsilon$.
Rectangle $C_i$ has side lengths $w_{i-1}$ and $w_i$, both of which are strictly less than $\epsilon$.

Let $k=k(r)$ denote the smallest possible number of beacons that can guard $P_r$. We will
say that a sequence of beacons $b_1,\dots, b_k$ is a {\it greedy placement} if 
$s(b_1) \prec \dots \prec s(b_k)$, and the
sequence $s(b_1),\dots,s(b_k)$ is maximum in lexicographical order. 
So intuitively, we obtain the greedy placement
by pushing the beacons as far as possible from the origin $v_0$ of the spiral, and giving
priority to the earliest beacons in the sequence.
Clearly, $b_1$ must be placed in $C_2$.
We then observe the following for $b_2, \ldots, b_{k-1}$.
\begin{lemma} \label{lem:spiral_sub}
 For $2\leq i\leq k-1$, the $i$-th beacon $b_i$ in a greedy placement for $P_r$
 is  in $A_{3i-1}$.
\end{lemma}

\begin{figure}[tb]
\centering
\includegraphics[width=\textwidth]{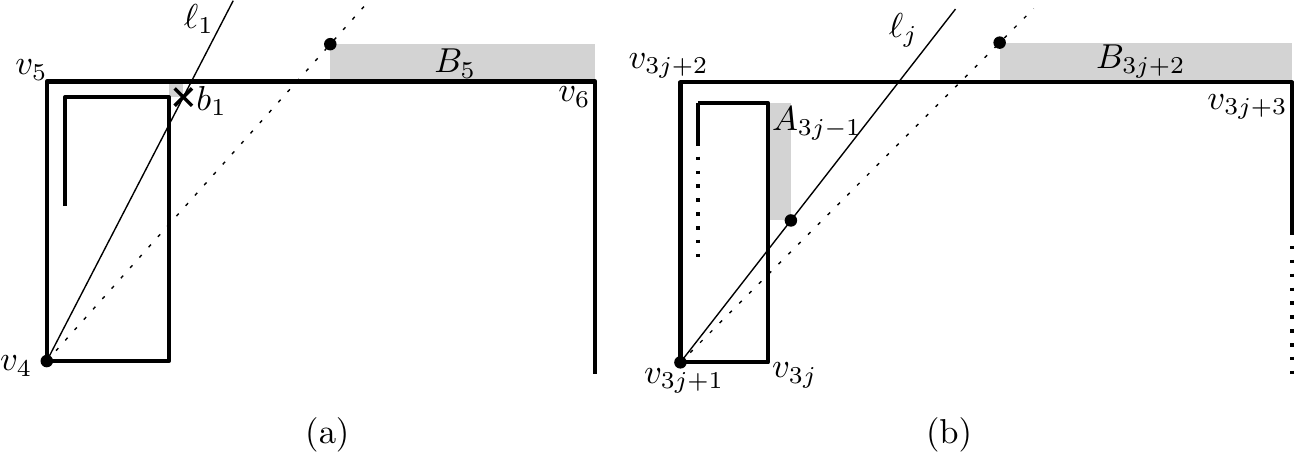}
\caption{Proof of Lemma~\ref{lem:spiral_sub}.}
\label{fig:spirals2}
\end{figure}

\begin{proof}
We prove the lemma by induction on $i$.
We first verify the lemma for $b_2$.
Without loss of generality, we assume that the edge $v_1v_2$ is a top reflex edge.
Let $\ell_1$ be the line through $v_4$ and $b_1$.
Observe that $b_1$ attracts all points in $C_3$, but not all of those in $C_4$.
More precisely, $b_1$ attracts those in $C_4$ below $\ell_1$ but miss those above $\ell_1$.
Hence, $b_2$ must be placed on $\ell_1$ to cover the points in $C_4$ above $\ell_1$.
For our purpose, we compare the slopes of $\ell_1$ and any line through $v_4$ and a point in $B_5$.
(See \figurename~\ref{fig:spirals2}a.)
Recall that $\rho > 2 + (r / 2 + 2)\epsilon$.
The slope of $\ell_1$ is at least
\[ \frac{\rho + \epsilon}{1+ 2\epsilon} > 1, \]
since $a_1 = 1$, $a_2 = \rho + \epsilon$, $a_3 = 1 + \epsilon$, and
the width $w_3$ of $C_3$ is at most $\epsilon$.
On the other hand, the slope of any line through $v_4$ and a point in $B_5$ is at most
\[ \frac{\rho + 3\epsilon}{\rho^2 + 2\epsilon} < 1,\]
since $a_4 = \rho + 2\epsilon$, $a_5 = \rho^2 + 2\epsilon$, and the height $w_5$ of $A_5$
is at most $\epsilon$.
This implies that $\ell_1$ cannot intersect $B_5$.
Thus, if $b_2\in B_5$, then $b_2$ fails to attract some points near $v_4$ and above $\ell_1$,
similarly as in \figurename~\ref{fig:spirals}a, so $b_2$ must lie in $A_5$.

For the inductive step, assume that $j \in \{2, \ldots, k-2\}$ and $b_j$  lies in $A_{3j-1}$.
If $r \leq 3j$, then $b_j$ attracts all points in $C_{3j} \cup A_{3j} \cup B_{3j}$,
that is, $b_j$ must be the last beacon in the greedy placement.
But this is not the case for our assumption that $j \leq k-2$.
We thus have $r \geq 3j + 1$.
Then, $C_{3j+1}$ cannot be completely covered by $b_j$, so the next beacon $b_{j+1}$
must cover $C_{3j+1}$ partially.
More precisely, $b_{j+1}$ must lie on the line $\ell_j$ through $v_{3j+1}$ and $b_j$.
Without loss of generality, we assume that the edge $v_{3j-1}v_{3j}$ is a right reflex edge.
Also, note that $r \geq 3j + 2$, since otherwise placing $b_{j+1}$ completes the greedy placement
and thus $k = j+1$, which is not the case.

Let $h(m) := \lfloor \frac{m}{2} \rfloor$, and
let $L$ be the positive number such that $a_{3j-2} = L + h(3j-2)$.
Recall that $a_{3j-2}$ is the length of edge $v_{3j-2}v_{3j-1}$.
We then have $a_{3j-1} = \rho L + h(3j-1)$, $a_{3j} = L + h(3j)$, $a_{3j+1} = \rho L + h(3j+1)$,
and $a_{3j+2} = \rho^2 L + h(3j+2)$.
See \figurename~\ref{fig:spirals2}(b).
Similarly to the above argument,
the slope of $\ell_js$ is at least
\[ \frac{a_{3j-1}/2}{a_{3j} + \epsilon} = \frac{(\rho L + \lfloor \frac{3j-1}{2} \rfloor)/2}
  {L+ (\lfloor \frac{3j}{2} \rfloor + 1)\epsilon} > 1,\]
since $b_{j-1}$ lies in $A_{3j-1}$ and
$\rho > 2+ (\frac{r}{2} + 2)\epsilon > 2 + (\frac{3j}{2} + 2)\epsilon / L$.
On the other hand, the slope of any line through $v_{3j+1}$ and any point in $B_{3j+2}$ is at most
\[ \frac{a_{3j+1} + \epsilon}{a_{3j+2}/2} < 1,\]
since $\rho > 2+ (\frac{r}{2} + 2)\epsilon >  2 + (\frac{3j+2}{2} + 2)\epsilon / \rho L$,
This implies that the next beacon $b_{j+1}$ also must be placed in $A_{3j+2}$. 
\end{proof}

The main result of this section follows:
\begin{lemma} \label{lem:coverage_lower_bound}
 The spiral $P_r$ defined above cannot be guarded by less than
 $\lceil \frac{r}{3} \rceil = \lfloor \frac{n}{6} \rfloor$ beacons,
 where $n$ denotes the number of vertices of $P_r$.
\end{lemma}


\subsection{Proof of the upper bound for coverage}\label{subsec:upper_bound_coverage}

In this section, we prove the matching upper bound $\lceil \frac{r}{3} \rceil$.
Our proof is by induction on $r$.
The following lemma handles the base case.
\begin{lemma} \label{lem:coverage_base}
 Any rectilinear polygon $P$ with at most three reflex vertices can be guarded by
 a single beacon.  Moreover, the beacon  can be placed at a reflex vertex of $P$,
 provided that $P$ has at least one.
\end{lemma}
\begin{proof}
Let $P$ be a rectilinear polygon with $r = r(P) \leq 3$.
If $P$ has no reflex edge, then $P$ is $xy$-monotone by Observation~\ref{obs:monotone},
and thus a beacon placed at any point in $P$ guards $P$
by Theorem~\ref{thm:kernel}.

Observe that $P$ has at most two reflex edges,
and this is possible only when its three reflex vertices are consecutive.
Assume that this is the case.
Then, the two reflex edges $e$ and $e'$ of $P$ must be adjacent and share a reflex vertex $v$.
Hence, the region $\RKernel(P) = H_e \cap H_{e'}$ forms a cone with the right angle at apex $v$.
Since $v$ is obviously contained in $P$,  Theorem~\ref{thm:kernel}
implies that $v\in \Kernel(P)$, so placing a beacon at $v$ is sufficient to guard $P$.

Now, suppose that $P$ has exactly one reflex edge $e$.
Then, $\RKernel = H_e$ forms a half-plane, and $e$ is contained in $P$.
We place a beacon $b$ at a reflex vertex incident to $e$.
Since $e \subset \Kernel(P)$ by Theorem~\ref{thm:kernel},
$b$ guards $P$.
\end{proof}

If $r \geq 4$, then we will partition $P$ using cuts.
A {\it normal} cut is a cut that is not incident to any vertex of $P$.
We will try to use normal cuts as often as possible,
as the new edges in the subpolygons created by a normal cut are convex,
and thus by Observation~\ref{obs:convex_edge},
these subpolygons can be handled separately.
We are more interested in normal cuts with an additional property:
A normal cut $c$ in $P$ is called \emph{safe} if
$r(P^+_c) \geq 1$, $r(P^-_c) \geq 1$, and
$\lceil \frac{r(P^+_c)}{3} \rceil + \lceil \frac{r(P^-_c)}{3} \rceil = \lceil \frac{r}{3} \rceil$.

A normal cut $c$ in $P$ is called an \emph{$m$-cut} if $r(P^-_c) \equiv m \pmod 3$.
We will abuse notation and write $a \equiv b$ instead of $a \equiv b \pmod 3$.
\begin{lemma} \label{lem:safe_cut}
 Let $c$ be any normal cut in $P$ such that $r(P^+_c) \geq 1$ and $r(P^-_c) \geq 1$.
 Then, $c$ is safe if and only if either $r \equiv 1$, or
 $c$ is a $0$-cut or an $r$-cut.
\end{lemma}
\begin{proof}
Let $r^+ := r(P^+_c)$ and $r^- := r(P^-_c)$.
Note that $r = r^+ + r^-$.
Since we assume that $r^+ \geq 1$ and $r^- \geq 1$,
the cut $c$ is safe if and only if
$\lceil \frac{r^+}{3} \rceil + \lceil \frac{r^-}{3} \rceil = \lceil \frac{r}{3} \rceil$.

First, suppose that $r \equiv 1$.
If $c$ is  a $0$-cut or a $1$-cut, 
then $r^+ \equiv 0$ or $r^- \equiv 0$, which implies 
$ \left\lceil \frac{r^+}{3} \right\rceil + \left\lceil \frac{r^-}{3} 
\right\rceil  = \left\lceil \frac{r^+ + r^-}{3} \right\rceil    
= \left\lceil \frac{r}{3}\right\rceil$.
If on the other  hand $c$ is a $2$-cut, then we have $r^+ \equiv r^- \equiv 2$, and
thus
\[ \left\lceil \frac{r^+}{3} \right\rceil + \left\lceil \frac{r^-}{3} \right\rceil = \frac{r^+ + 1}{3} + \frac{r^- + 1}{3}
   = \frac{r^+ + r^- + 2}{3} = \frac{r+2}{3} = \left\lceil \frac{r}{3} \right\rceil.\]

Now we assume that $r \not\equiv 1$, that is, $r \equiv 0$ or $r \equiv 2$.
Then $c$ is a $0$-cut or an $r$-cut if and only if  $r^+ \equiv 0$ or $r^- \equiv 0$, which 
is equivalent to 
$ \left\lceil \frac{r^+}{3} \right\rceil + \left\lceil \frac{r^-}{3} \right\rceil  = \left\lceil \frac{r^+ + r^-}{3}
 \right\rceil    = \left\lceil \frac{r}{3}\right\rceil$.
\end{proof}

If there is a safe cut $c$ in $P$ when $r \geq 4$,
then one can partition $P$ into two subpolygons $P^+_c$ and $P^-_c$ by $c$
and attain the target bound $\lceil \frac{r}{3}\rceil$ on the number of beacons
by handling each subpolygon by our induction hypothesis.
But, this is not always the case:
there exist rectilinear polygons $P$ that do not admit a safe cut
when $r \not\equiv 1$.

\begin{lemma} \label{lem:1-cut}
 Suppose that $P$ admits no safe cut and $r\geq 5$.
 Then, there exists a horizontal normal cut $c$ in $P$ such that
 either $c$ is a $1$-cut with $r(P^-_c) \geq 4$, or $c$ is a $2$-cut with $r(P^+_c) \geq 4$.
\end{lemma}
\begin{proof}
Since $P$ admits no safe cut, we know that $r \equiv 0$ or $r \equiv 2$
by Lemma~\ref{lem:safe_cut}.
We separately handle the cases where $r \equiv 0$ or $r\equiv 2$.

\begin{figure}[tb]
\centering
\includegraphics[width=\textwidth]{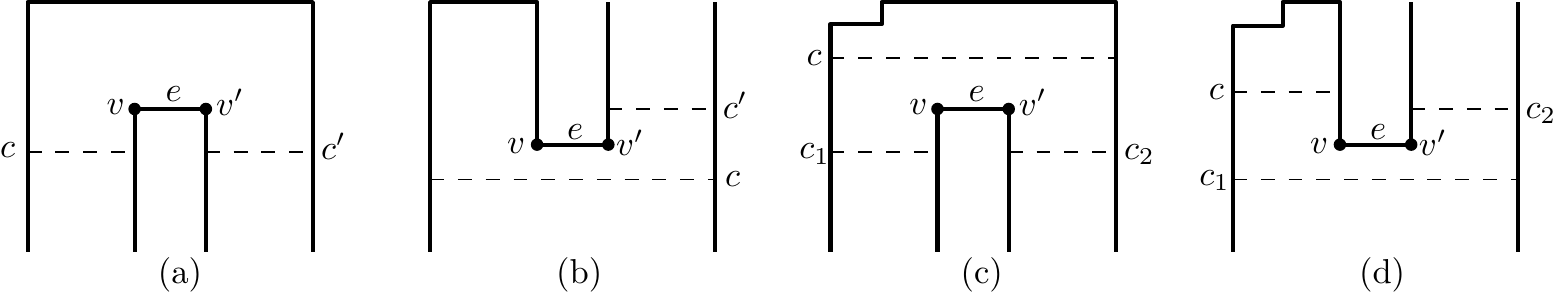}
\caption{Proof of Lemma~\ref{lem:1-cut}}
\label{fig:1-cut}
\end{figure}

First, assume that $r\equiv 2$.
In this case, we show a stronger claim (\figurename~\ref{fig:1-cut}a):
\begin{quote}
\textit{If $r\equiv 2$, then for any top convex edge $e_0$ of $P$,
the first reflex vertex $v$ below $e_0$ is not an endpoint of any horizontal reflex edge of $P$.}
\end{quote}
This automatically proves the lemma: A normal cut $c$ just below $v$ is a $1$-cut and
$r(P^-_c) = r - 1 \geq 4$.
Suppose to the contrary that it is not the case,
so there is a top convex edge $e_0$ of $P$ such that
the first reflex vertex $v$ below $e_0$ is an endpoint of a reflex edge $e$ of $P$.
Let $v'$ be the other endpoint of $e$.
There are two cases: Either $e$ is a top reflex edge (\figurename~\ref{fig:1-cut}b), 
or a bottom reflex edge (\figurename~\ref{fig:1-cut}c).

Consider the first case, where $e$ is a top reflex edge.
Let $c$ and $c'$ be normal cuts just below $v$ and $v'$, respectively.
We cannot have $r(P^-_c)=0$ as it would imply that  $r(P^+_{c'}) \equiv 0$, an hence
$c'$ would be a safe cut. Similarly, we have $r(P^-_c) \neq 0$.
So by Lemma~\ref{lem:safe_cut}, as there is no safe cut in $P$, 
both $c$ and $c'$ must be $1$-cuts, and thus $r(P^-_c) + r(P^-_{c'}) \equiv 2$.
It implies $r=r(P^-_c) + r(P^-_{c'})+2 \equiv 1$, a contradiction.

Now, consider the latter case where $e$ is a bottom reflex edge.
See \figurename~\ref{fig:1-cut}c.
Let $c$ be a normal cut just below $e$ and let $c'$ be a normal cut just above $v'$.
Since $r(P^-_c) + r(P^+_{c'}) = r - 2 \equiv 0$ and since
there is no safe cut in $P$, both $c$ and $c'$ must be $1$-cuts,
so $r(P^-_c) + r(P^+_{c'}) = r(P^-_c) + ( r - r(P^-_{c'})) \equiv 2$, a contradiction.
Thus, our claim for the case $r \equiv 2$ is true.

Assume now that $r\equiv 0$, and thus $r \geq 6$. By our assumption that $P$ has not safe 
cut, there is no $0$-cut $c'$ with $r(P^+_{c'}) \geq 1$ and $r(P^-_{c'}) \geq 1$.
Now suppose that the lemma is false.
Then, we have $r(P^-_c) = 1$ for any $1$-cut $c$ in $P$, and
$r(P^+_c) = 1$ for any $2$-cut $c$ in $P$.
Pick any $2$-cut $c$ in $P$.
If there is no $2$-cut in $P$, then we rotate $P$ by $180^\circ$
so that every $1$-cut is transformed into a $2$-cut.
Note that $r(P^+_c) = 1$ and $r(P^-_c) = r - 1 \geq 5$.
Let $v$ be the first reflex vertex below $c$.
If $v$ is not incident to a horizontal reflex edge,
then a normal cut $c’$ just below $v$ is a $1$-cut with $r(P^-_{c’}) \geq 4$,
a contradiction.
Thus, $v$ is incident to a horizontal reflex edge $e$.

There are two cases: either $e$ is a top reflex edge or a bottom reflex edge.
Assume that $e$ is a top reflex edge.
Let $c_1$ and $c_2$ be normal cuts just below $v$ and $v'$, respectively,
where $v'$ is the other vertex incident to $e$.
(See \figurename~\ref{fig:1-cut}c.)
For any $i \in \{1, 2\}$, if $c_i$ is a $2$-cut, then $r(P^+_{c_i}) \equiv 1$
and $r(P^+_{c_i}) \geq 4$, a contradiction.
Thus, neither $c_1$ nor $c_2$ can be a $2$-cut.
Moreover, since $r(P^-_{c_1}) + r(P^-_{c_2}) = r - 3 \equiv 0$,
we must have $r(P^-_{c_1}) \equiv r(P^-_{c_2}) \equiv 0$.
Since $P$ admits no safe cut, it implies that 
$r(P^-_{c_1}) = r(P^-_{c_2}) = 0$, and hence $r = 3$, a contradiction
to the assumption that $r\geq 6$.

Assume that $e$ is a bottom reflex edge.
Let $c_1$ be a normal cut just below $v$
and $c_2$ be a normal cut just above $v'$.
(See \figurename~\ref{fig:1-cut}c.)
In this case, $c_1$ cannot be a $2$-cut since $r(P^+_{c_1}) \geq 3$,
while $c_2$ cannot be a $1$-cut since $r(P^-_{c_2}) \geq 3$.
On the other hand, we have $r(P^-_{c_1}) + r(P^+_{c_2}) = r - 3 \equiv 0$,
and thus $r(P^-_{c_1}) \equiv r(P^-_{c_2})$.
So, we must have $r(P^-_{c_1}) \equiv r(P^-_{c_2}) \equiv 0$.
Since $P$ has no safe cut, it implies that  $r(P^-_{c_1}) = r(P^+_{c_2}) = 0$,
and thus $r = 3$, a contradiction.
\end{proof}

Now, we are ready to prove the main result of this section.
\begin{lemma} \label{lem:covering}
 Let $P$ be a simple rectilinear polygon $P$ with $n \geq 6$ vertices and
 $r \geq 1$ reflex vertices.
 Then, $\lfloor \frac{n}{6}\rfloor = \lceil \frac{r}{3} \rceil$ beacons are sufficient
 to guard $P$.  Moreover, all these beacons can be placed at reflex vertices of $P$.
\end{lemma}
\begin{proof}
Our proof is by induction on $r$.
The base case where $r \leq 3$ is already handled by Lemma~\ref{lem:coverage_base},
so  we assume that $r \geq 4$.
If $P$ admits a safe cut $c$, then we partition it into two subpolygons $P^+_c$ and $P^-_c$,
and  we handle each subpolygon recursively. 
Our guarding set for $P$ is the union of the guarding sets
for $P^+_c$ and $P^-_c$. As $c$ is safe, the total number of beacons we place is at most
 \[ \left\lceil \frac{r(P^+_c)}{3} \right\rceil + \left\lceil \frac{r(P^-_c)}{3} \right\rceil
    = \left\lceil \frac{r}{3} \right\rceil.\]
These beacons indeed guard $P$, because
$c$ is a convex edge of each subpolygon,
and thus by Observation~\ref{obs:convex_edge}, no beacon attraction path hits $c$.

Now we assume that $P$ admits no safe cut.
Then, by Lemma~\ref{lem:safe_cut}, we have $r \not\equiv 1$, so $r \geq 5$,
and there is no $0$-cut or $r$-cut with at least one reflex vertex on each side.
Consider the set $C$ of all $1$-cuts $c'$ in $P$ with $r(P^-_{c'}) \geq 4$.
By Lemma~\ref{lem:1-cut}, we may assume that $C$ is nonempty:
If $r \equiv 2$, then it immediately follows that $C \neq \emptyset$, 
and  if $r \equiv 0$ and $C = \emptyset$,
then we rotate $P$ by $180^\circ$.

Pick a $1$-cut $c \in C$ such that $r(P^-_c)$ is minimum.
Let $v$ be the first reflex vertex of $P$ below $c$.
If $v$ is not an endpoint of a horizontal reflex edge, then a normal cut $c'$ just below $v$
is a $0$-cut with $r(P^+_{c'}) \geq 1$ and $r(P^-_{c'}) \geq 1$,
so it is a safe cut, a contradiction to the assumption that $P$ admits no safe cut.
Hence, $v$ must be an endpoint of a horizontal reflex edge $e$.
We have two cases: Either $e$ is a top reflex edge, or a bottom reflex edge.

\begin{figure}[tb]
\centering
\includegraphics[width=\textwidth]{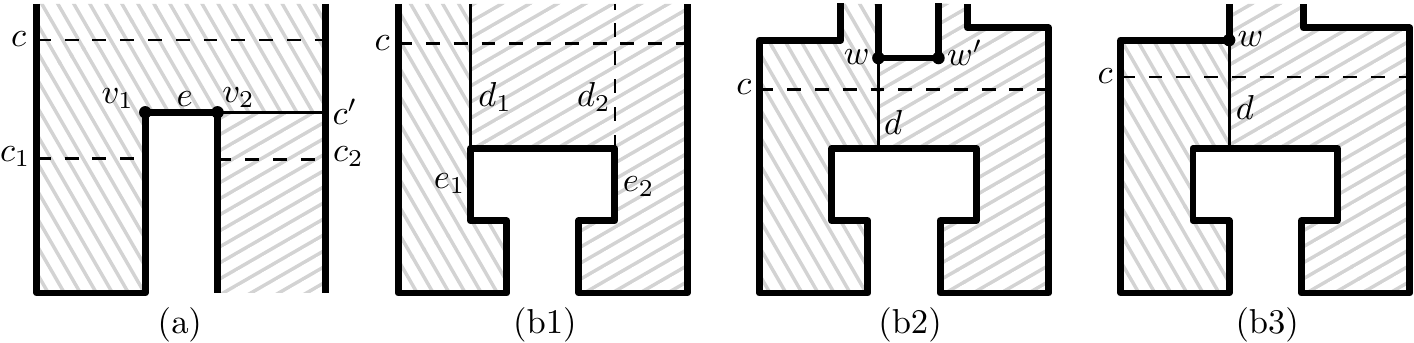}
\caption{Proof of Theorem~\ref{thm:covering} when $e$ is a top reflex edge.
(a) When $c_1$ is a $0$-cut and $c_2$ is a $2$-cut.
(b1)--(b3) When $c_1$ and $c_2$ are $1$-cuts.}
\label{fig:covering_top}
\end{figure}

\paragraph{When $e$ is a top reflex edge.}
Assume the former case where $e$ is a top reflex edge.
Let $v_1$ and $v_2$ be the left and right endpoint of $e$, respectively, and
$c_1$ and $c_2$ be normal cuts just below $v_1$ and $v_2$,
respectively.
Also, let $m_1,m_2\in \{0, 1, 2\}$ be such that
$c_1$ is an $m_1$-cut and $c_2$ is an $m_2$-cut.
We then observe that $m_1 + m_2 + 2 \equiv 1$, that is, $m_1 + m_2 \equiv 2$.
We treat separately the two possible cases: 
Case (a), where $(m_1,m_2)=(0, 2)$ or $(2, 0)$, and Case (b), where $(m_1,m_2)=(1,1)$.
(See \figurename~\ref{fig:covering_top}.)
\begin{enumerate}[(a)] 
\item Assume that $(m_1, m_2) = (0, 2)$, so $c_1$ is a $0$-cut and $c_2$ is a $2$-cut.
 (The case where $(m_1, m_2) = (2, 0)$ is symmetric, and can be handled in the same way.)
 As $c_1$ is not a safe cut, we have $r(P^-_{c_1}) = 0$.
 Consider the horizontal cut $c'$ at $v_2$.
 (See \figurename~\ref{fig:covering_top}a.)
 We have $r(P^-_{c'}) \equiv 2$, and $r(P^+_{c'}) + r(P^-_{c'}) = r - 1$.
 We then place beacons in $P^+_{c'}$ and $P^-_{c'}$ separately and recursively.
 The total number of beacons placed is at most
 \[ \left\lceil \frac{r(P^+_{c'})}{3} \right\rceil + \left\lceil \frac{r(P^-_{c'})}{3} \right\rceil
   = \left\lceil \frac{r(P^+_{c'})}{3} \right\rceil + \frac{r(P^-_{c'}) + 1}{3}
   = \left\lceil \frac{r(P^+_{c'})+r(P^-_{c'})+1}{3} \right\rceil = \left\lceil \frac{r}{3} \right\rceil.\]

 We still need to make sure that these beacons indeed guard $P$,
 Our induction hypothesis implies that all the beacons are placed at reflex vertices of $P$.
 So there is no beacon placed below the horizontal cut at $v_1$,
 and thus, by Observation~\ref{obs:convex_edge}, no beacon attraction path in $P^+_{c'}$ hits $c'$.
 Again by Observation~\ref{obs:convex_edge}, the beacons placed in $P^-_{c'}$ indeed guard
 the region $P^-_{c'}$ in $P$ since $c'$ is a convex edge of $P^-_{c'}$.
 This ensures that the beacons we placed separately in $P^+_{c'}$ and $P^-_{c'}$
 indeed guard $P$.

\item We now consider the case where $(m_1, m_2) = (1, 1)$.
 Then, we have $r(P^-_{c_1}) = r(P^-_{c_2}) = 1$ by our choice of $c$.
 Let $e_1$ and $e_2$ be the edges other than $e$  incident to $v_1$ and $v_2$, respectively.
 If $e_1$ is not a reflex edge, then a beacon placed at $v_2$ guards $P^-_c$
 since its kernel $\Kernel(P^-_c)$ is nonempty.
 The other part $P^+_c$ can be guarded by at most $\lceil \frac{r(P^+_c)}{3} \rceil$
 guards placed on reflex vertices of $P^+_c$.
 Since $c$ is a normal cut, these $1+\lceil \frac{r(P^+_c)}{3} \rceil$ beacons
 together guard $P$.
 As $r(P^+_c) = r - 4$, the number of beacons is bounded by
 \[ 1+\left\lceil \frac{r(P^+_c)}{3} \right\rceil = \left\lceil \frac{r(P^+_c)+3}{3} \right\rceil
  \leq \left\lceil \frac{r}{3} \right\rceil,\]
 as desired.
 The case where $e_2$ is not a reflex edge   can be handled symmetrically
 by placing a beacon at $v_1$.

 Thus, we now assume that both $e_1$ and $e_2$ are reflex edges.
 Let $d_1$ and $d_2$ be the vertical cuts at $v_1$ and $v_2$, respectively.
 We handle three subcases separately:
 (i) either $r(P^-_{d_i}) \equiv 0$ or $r(P^+_{d_i}) \equiv 0$ for some $i\in \{1,2\}$,
 (ii) $r \equiv 2$ and $r(P^-_{d_i}) \equiv r(P^+_{d_i}) \equiv 2$ for each $i \in \{1,2\}$,
 or (iii) $r\equiv 0$ and $r(P^-_{d_i}) \equiv r(P^+_{d_i}) \equiv 1$ for each $i \in \{1,2\}$.
 If we are not in Case (i), we have either Case (ii) or (iii).
 So these three cases cover all possible situations.
 \begin{enumerate}[(i)]
 \item Without loss of generality, we assume that $r(P^-_{d_1}) \equiv 0$ or $r(P^+_{d_1}) \equiv 0$.
 In this case,  we handle $P^-_{d_i}$ and $P^+_{d_i}$ recursively.
 The union of the two guarding sets of $P^-_{d_i}$ and $P^+_{d_i}$ guards $P$,
 as all these beacons are placed at reflex vertices.
 (See \figurename~\ref{fig:covering_top}(b1).)
 The number of  beacons is at most  $\lceil (r(P^-_{d_1})+r(P^+_{d_1}))/3 \rceil
 = \lceil \frac{r}{3} \rceil$,
 since $r(P^-_{d_1}) \geq 1$ and $r(P^+_{d_1}) \geq 1$.

 \item 
 Assume that $r \equiv 2$ and $r(P^-_{d_i}) \equiv r(P^+_{d_i}) \equiv 2$ for each $i=1,2$.
 We partition $P$ into $P^+_{d_1}$ and $P^-_{d_1}$ and handle them recursively.
 The total number of beacons is at most
 \[\left\lceil \frac{ r(P^+_{d_1})}{3} \right\rceil + \left\lceil \frac{r(P^-_{d_1})}{3} \right\rceil
   = \frac{r(P^+_{d_1})+1}{3} + \frac{r(P^-_{d_1})+1}{3} = \frac{r + 1}{3} =
   \left\lceil \frac{r}{3} \right\rceil.\]
 These beacons guard $P$, as they are all placed at reflex vertices of $P$.
 \item
 The remaining case is when $r \equiv 0$ and $r(P^-_{d_i}) \equiv r(P^+_{d_i}) \equiv 1$
 for each $i=1, 2$.
 Then a vertical cut just to the left of $d_1$ is a $0$-cut,
 and a vertical cut just to the right of $d_2$ is a $0$-cut.
 Since $P$ admits no safe cut, this implies that $r(P^-_{d_1}) = r(P^+_{d_2}) = 1$.

 Let $w$ be the first reflex vertex above $c$, and let $d$ be the vertical cut at $w$.
 We consider two cases: Either $w$ is incident to a bottom reflex edge, or not.
 In the former case, let $w'$ be the other endpoint of the bottom reflex edge.
 (See \figurename~\ref{fig:covering_top}(b2).)
 Since $c$ is a $1$-cut and $r\equiv 0$, one of the two horizontal cuts just above $w$ and $w'$
 must be either a $0$-cut or a $2$-cut.
 Without loss of generality, assume that the cut above $w$ is  either a $0$-cut or a $2$-cut.
 As $r \equiv 0$, it means that the number of reflex vertices above this cut is $0$ or $1$
 modulo 3, and hence $r(P_d^-) \equiv 0$ or $2$.

 In the latter case, where $w$ is not incident to a bottom reflex edge,
 we may assume without loss of generality that the horizontal edge incident to $w$
 is to the left of $w$. (See \figurename~\ref{fig:covering_top}(b3).)
 So we still have $r(P_d^-) \equiv 0$ or $2$.

 In both cases, we partition $P$ by the vertical cut $d$ at $w$ into $P^-_{d}$ and $P^+_{d}$.
 The endpoint of $d$ other than $w$ always lies on the reflex edge $e$,
 since $r(P^-_{d_1}) = r(P^+_{d_2}) = 1$.
 Since $r(P^-_{d}) \equiv 0$ or $2$, and $r(P^+_{d}) \equiv r(P^-_{d}) + 1$, and
 $e \equiv 0$,
 we have either $r(P^-_{d}) \equiv 0$ or $r(P^+_{d}) \equiv 0$.
 We handle $P^-_{d}$ and $P^+_{d}$, separately, and recursively.
 Then the total number of beacons placed in $P$ is at most
 \[ \left\lceil \frac{r(P^-_{d}) }{3} \right\rceil + \left\lceil \frac{r(P^+_{d}) }{3} \right\rceil
  = \left\lceil \frac{r(P^-_{d}) + r(P^+_{d})}{3} \right\rceil = \left\lceil \frac{r}{3} \right\rceil.\]
 We still need  to verify that these beacons guard $P$.
 As $r(P^-_{d_1})=1$, by our induction hypothesis, there must be a beacon at one of the 
 endpoints of $e_1$.
 Such a beacon attracts all points to the left of $d_1$, and thus
 the region $P^-_{d}$ is covered by the beacons in $P$.
 Since $d$ is a convex edge of $P^+_{d}$, no attraction path hits $d$ inside $P^+_{d}$
 either.
 \end{enumerate}

\end{enumerate}
This completes the proof for the case where $e$ is a top reflex edge.

\begin{figure}[tb]
\centering
\includegraphics[width=\textwidth]{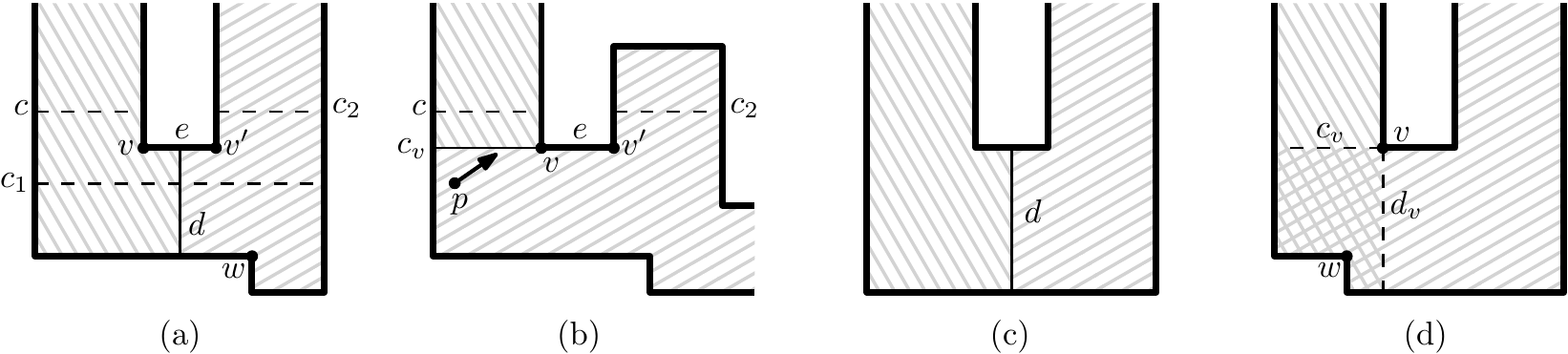}
\caption{Proof of Theorem~\ref{thm:covering} when $e$ is a bottom reflex edge.
(a) When $r \equiv 2$.
(b) When $r \equiv 0$ and $c_2$ is a $0$-cut.
(c) When $r \equiv 0$, $c_1$ is a $0$-cut, and $c_2$ is a $1$-cut.
(d) When $r \equiv 0$, $c_1$ is a $1$-cut, and $c_2$ is a $2$-cut. }
\label{fig:covering_bottom}
\end{figure}

\paragraph{When $e$ is a bottom reflex edge.}
We now assume that $e$ is a bottom reflex edge.
Let $v'$ be the other endpoint of $e$. (See \figurename~\ref{fig:covering_bottom}.)
Without loss of generality, we assume that $v$ is to the left of $v'$.
Let $c_1$ be a normal horizontal cut just below $e$ and
$c_2$ be a normal cut just above $v'$.
Also, let $m_1, m_2 \in \{0,1,2\}$ be such that $c_1$ is an $m_1$-cut
and $c_2$ is an $m_2$-cut.
So we have $(r-1) + (r-m_2) + 2 \equiv (r - m_1)$,
that is, $m_2 - m_1 \equiv r + 1$.
Recall that $c\in C$ has been chosen to be a $1$-cut
with the smallest value of $r(P^-_c)$ among all $1$-cuts $c'$ of $P$ with $r(P^-_{c'}) \geq 4$.

We first assume that $r \equiv 2$, and thus
 $m_1 \equiv m_2$.
 If $m_1 \equiv 2$, then $c_1$ is a safe cut.
 Similarly, if $m_2 \equiv 0$, then $c_2$ is a safe cut.
 Since $P$ admits no safe cut, we have that $m_1 \equiv m_2 \equiv 1$.
 So by our choice of $c$, we have $r(P^-_{c_1}) = 1$.
 Let $w$ be the unique reflex vertex in $P^-_{c_1}$.
 We pick any normal vertical cut $d$ at any point on $e$.
 (See \figurename~\ref{fig:covering_bottom}a.)
Since $d$ is not a safe cut, $d$ must be a $1$-cut, that is, $r(P^-_d) \equiv 1$.
On the other hand, we have $r(P^-_d) = r(P^+_c) + 2$ if $w \in P^-_d$,
and $r(P^-_d) = r(P^+_c) + 1$ if $w \notin P^-_d$.
As $c$ is a $1$-cut, it implies that either $r(P^-_d) \equiv 2$ or $r(P^-_d) \equiv 0$, a contradiction, so $r \not\equiv 2$.

We hence have $r \equiv 0$ and $m_2 \equiv m_1 + 1$. We first rule out $m_2 \equiv 0$.
So we assume, for sake of contradiction, that $m_2 \equiv 0$. We must have $r(P_{c_2}^+) = 0$, 
as otherwise $c_2$ would be safe cut.
We make a cut $c_v$ to the left of $v$.  (See \figurename~\ref{fig:covering_bottom}b.)
As $c_v$ is a 1-cut, we have $r(P_{c_v}^-) \equiv 0$, and $r(P_{c_v}^+) \equiv 2$. We recursively
construct guarding sets of beacons for $r(P_{c_v}^-)$ and $r(P_{c_v}^+)$. We claim that these
beacons together guard $P$. Indeed, the only way a point $p$ may not be covered would
be that $p$ lies in $P^-_{c_v}$ and its attraction path crosses $c_v$ from below.
But then $p$ would be attracted, within $P^-_{c_v}$, by the same beacon as $v'$. This
is impossible because there is no reflex vertex above the cut $c_2$, and by our induction
hypothesis, beacons are placed at reflex vertices. To complete the proof for
this case, we only need to bound the number of beacons we placed. It is at most
\[ \left\lceil \frac{r(P^-_{c_v}) }{3} \right\rceil + \left\lceil \frac{r(P^+_{c_v}) }{3} \right\rceil
  = \frac{r(P^-_{c_v}) }{3} +  \frac{r(P^+_{c_v})+1 }{3} = \frac r 3.
\]	

We now rule out the case where $(m_1, m_2) = (0, 1)$. 
If it were the case, then we would have $r(P^-_{c_1}) = 0$ since $P$ has no safe cut.
We pick any normal vertical cut $d$ at any point on $e$.
Since $d$ is a $0$-cut and both sides of $d$ contain at least one reflex vertex,
it is a safe cut, which contradicts our assumptions.
(See \figurename~\ref{fig:covering_bottom}c.)

We thus have $(m_1, m_2) = (1, 2)$.
Note that $r(P^-_{c_1}) = 1$ by our choice of $c$.
Let $w$ be the unique reflex vertex in $P^-_{c_1}$.
(See \figurename~\ref{fig:covering_bottom}d.)
Then $w$ must be to the left of $v$, because
otherwise, 
there would be a normal vertical cut $d$ at a point on $e$
such that $w$ lies to the right of $d$, and
$d$ would be a safe cut.

Let $c_v$ be the horizontal cut at $v$ and $d_v$ be the vertical cut at $v$.
We then handle  the two subpolygons $P^-_{c_v}$ and $P^-_{d_v}$, separately, in a recursive way,
though they partially overlap.
Due to the overlap, the union of the set of beacons placed separately in each subpolygon
guard the whole polygon $P$.
Since $r(P^-_{c_v}) \equiv r(P^-_{d_v}) \equiv 0$ and
$r(P^-_{c_v}) + r(P^-_{d_v}) = r$,
the number of beacons we placed is at most
\[ \frac{r(P^-_{c_v})}{3}  + \frac{r(P^-_{d_v})}{3} = \frac{r(P^-_{c_v})+r(P^-_{d_v})}{3}  =
 \frac{r}{3} = \left\lceil \frac{r}{3} \right\rceil.\]
This completes the proof.
\end{proof}

Our last result is to show that in the worst case, monotone rectilinear polygons require fewer
beacons than simple rectilinear polygons. It matches the lower bound by Biro~\cite{b-bbrg-13}.
\begin{theorem}\label{thm:coverage_monotone}
 For any rectilinear monotone polygon $P$ with $n$ vertices, $r$ of which are reflex,
 $\lfloor \frac{n+4}{8} \rfloor = \lfloor \frac{r}{4} \rfloor + 1$ beacons are sufficient to guard $P$,
 and sometimes necessary.
\end{theorem}
\begin{proof}
Without loss of generality, we assume that $P$ is $x$-monotone.
Thus, $P$ has no vertical reflex edge by Observation~\ref{obs:monotone}.
Our proof is by induction on the number $r$ of reflex vertices.
If $P$ has at most one reflex edge $e$,
then we observe that any point on $e$ is contained in the kernel $\Kernel(P)$
by Theorem~\ref{thm:kernel}.
Thus, one beacon is sufficient to guard $P$.

Now, assume that $P$ has at least two reflex edges.
This implies that $r \geq 4$ since $P$ is $x$-monotone.
Let $v_1, v_2, \ldots, v_k$ be the right endpoints of the reflex edges
sorted from left to right.
Let $e_1$ and $e_2$ be the reflex edges that are incident to $v_1$ and $v_2$, respectively.
Let $c$ be the vertical cut at $v_2$.
We partition $P$ into $P^+_c$ and $P^-_c$ by $c$.
Then the left side subpolygon $P^-_c$ has at most one reflex edge $e_1$,
and thus can be guarded by a single beacon placed at any point on $e_1$.
The right side subpolygon $P^+_c$ has $r(P^+_c) = r - 4$ reflex vertices.
Thus, by induction, at most $\lfloor \frac{r-4}{4} \rfloor + 1$ beacons
can guard $P^+_c$.
The total number of beacons placed in $P$ is at most
 \[ 1 + \left\lfloor \frac{r-4}{4} \right\rfloor + 1 = \left\lfloor \frac{r}{4} \right\rfloor + 1,\]
as desired.

Finally, observe that cutting by $c$ always makes a new convex edge in $P^-_c$ and $P^+_c$
since there is no vertical reflex edge in $P$.
This implies that separately guarding $P^-_c$ and $P^+_c$ is sufficient to guard the whole $P$
by Observation~\ref{obs:convex_edge}.
\end{proof}

\section{Beacon-Based Routing} \label{sec:routing}

In this section, we give an improved upper bound for the beacon-based routing problem in a simple rectilinear polygon $P$ with $r$ reflex vertices. 

Our result in this section is as follows.
\begin{theorem}
\label{thm:routing}
 For any simple rectilinear polygon $P$ with $n\geq 4$ vertices, $r\geq 0$ of which are reflex,
 $\lfloor \frac{3n-4}{8}\rfloor-1=\lfloor \frac{3r}{4}\rfloor$ beacons are always sufficient
 to route between all pairs of points in $P$. There are simple rectilinear polygons in which $\lceil \frac{n}{4}\rceil-1=\lceil \frac{r}{2}\rceil$ beacons are necessary to route all point pairs for any $n\neq 6$.
\end{theorem}

\subsection{Proof of the lower bound for routing}
The polygon constructed by Biro et al.~\cite{bgikm-cccg-13} to show the lower bound $\lfloor \frac{r}{2}\rfloor(=\lfloor \frac{n}{4}\rfloor -1)$ is $x$-monotone. (See \figurename~\ref{fig:intro}a.) But we can construct non-monotone polygons for which $\lceil \frac{r}{2}\rceil (=\lceil \frac{n}{4}\rceil -1)$ beacons are necessary for any $r\neq 1$, i.e., $n\neq 6$. When $r=1$, $P$ is $xy$-monotone, so no beacon is needed.

We construct a spiral polygon $P_r$ with $r$ reflex vertices. The construction is similar to the one for the lower bound on the coverage problem, so we explain only the idea of the construction. For $r=0$, no beacon is needed, and for $r=2$, one beacon is necessary, so it suffices to construct $P_r$ for $r \geq 3$. 

\begin{figure}[tb]
\centering
\includegraphics[width=\textwidth]{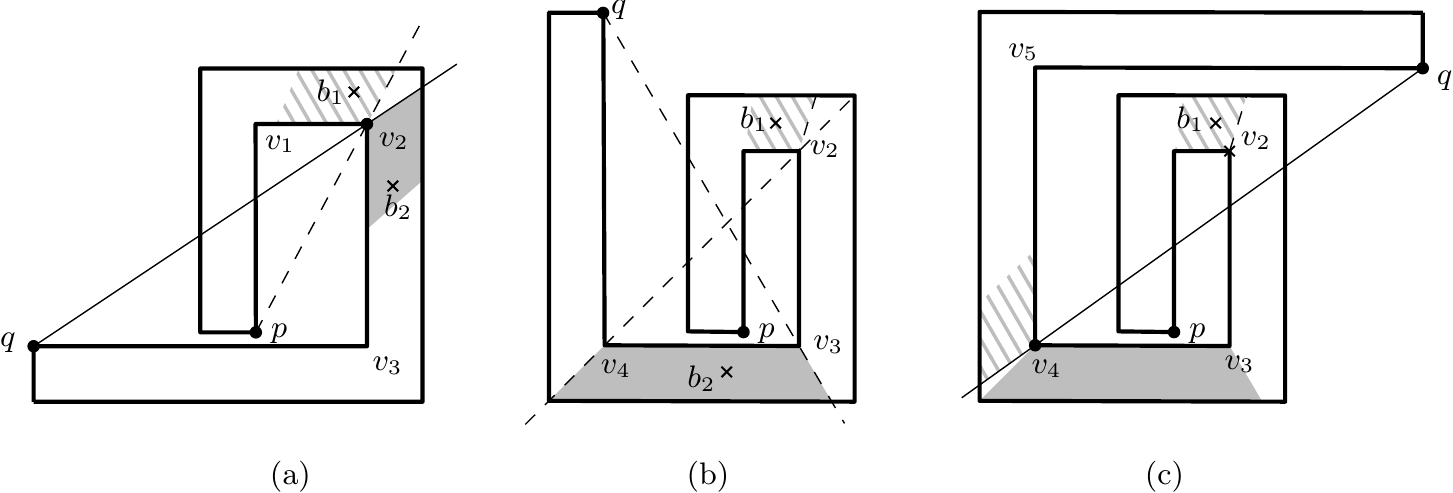}
\caption{Examples on the construction of $P_r$ for the lower bound on the routing problem. (a) $P_3$. (b) $P_4$. (c) $P_5$.}
\label{fig:lb_for_routing}
\end{figure}

As a base case for the construction, we consider $P_3$. See \figurename~\ref{fig:lb_for_routing}a. Let $p$ and $q$ be the end convex vertices of the spine of $P_r$. For a point $b_1$ to be attracted to $p$, $b_1$ must be in a polygonal region above the line $pv_2$, otherwise its beacon-based path to $p$ is blocked by the reflex edge connecting $v_2$ and $v_3$. If another point $b_2$ would be attracted to $q$, then $b_2$ must be in a polygonal region below the line $qv_2$. Since such two regions are disjoint except at $v_2$, one beacon is not sufficient. Note here that a beacon placed at $v_2$ is attracted neither to $p$ nor to $q$ because two paths toward $v_1$ and $v_3$ are both valid. So $P_3$ needs (at least) $\lceil \frac{3}{2}\rceil =2$ beacons. From $P_3$, we construct $P_r$ incrementally. 

Let us look at $P_4$. Two beacons are sufficient and necessary for $P_4$ as in \figurename~\ref{fig:lb_for_routing}b. We know that the second beacon $b_2$ should be placed in the region bounded by the line $qv_3$ and the line $v_4v_2$, otherwise $b_2$ cannot be attracted to $b_1$ due to the obstruction of the edge from $v_4$ to $q$. Then we can make $P_5$ with the endpoint $q$, as in \figurename~\ref{fig:lb_for_routing}c, such that the line $qv_4$ passes above the region where $b_2$ is placed.  This implies that $b_2$ cannot be attracted to $q$ because $b_2$ lies below the line $qv_4$, so the third beacon is necessary. It is not hard to check that this argument can be applied to construct $P_r$ from $P_{r-1}$ for any $r>3$. This shows the lower bound for the routing problem as follows.

\begin{lemma} \label{lem:coverage_lower_bound}
 The spiral $P_r$ defined above cannot be guarded by less than
 $\lceil \frac{r}{2} \rceil = \lceil \frac{n}{4}\rceil-1$ beacons,
 where $n$ denotes the number of vertices of $P_r$, i.e., $n=2r+4$, for any $r \neq 1$ or $n\neq 6$.
\end{lemma}

\subsection{Proof of the upper bound for routing}

We now explain how to place $\lceil \frac{3r}{4}\rceil$ beacons to route all pairs of points in $P$.

When $P$ is $xy$-monotone, we do not need to place any beacon in order to route between
a pair $(s,t)$ of points in $P$, since the target $t$ is regarded as a beacon.
We thus focus on the case where $P$ ifs not $xy$-monotone. 
Our approach is to cut $P$ by extending some of its edges, and
handle the resulting parts separately. More precisely,
for any reflex vertex $v$ of $P$  incident to an edge $e$,
we denote by $\cut_v(e)$ the cut obtained by extending $e$ through $v$.
So when  $e$ is horizontal, $\cut_v(e)$ is the horizontal cut at $v$, and
if $e$ is vertical, then $\cut_v(e)$ denotes the vertical cut at $v$.
The cut $\cut_v(e)$ splits $P$ into two subpolygons $P^+_{\cut_v(e)}$ and $P^-_{\cut_v(e)}$,
one of which does not contain $e$.
We call this subpolygon the \emph{pocket} of $e$ at $v$, denoted by $\poc_v(e)$.
For instance, if $e$ is a top reflex edge, then we have $\poc_v(e) = P^-_{\cut_v(e)}$.
So for any reflex edge $e$ with endpoints $v$ and $w$,
there are two cuts $\cut_v(e)$ and $\cut_w(e)$ extending $e$, 
and two pockets $\poc_v(e)$ and $\poc_w(e)$ of $e$.
(See \figurename~\ref{fig:monotone_pocket}a.)

Our solution to the routing problem is based on the following key lemma.

\begin{figure}[tb]
\centering
\includegraphics[width=\textwidth]{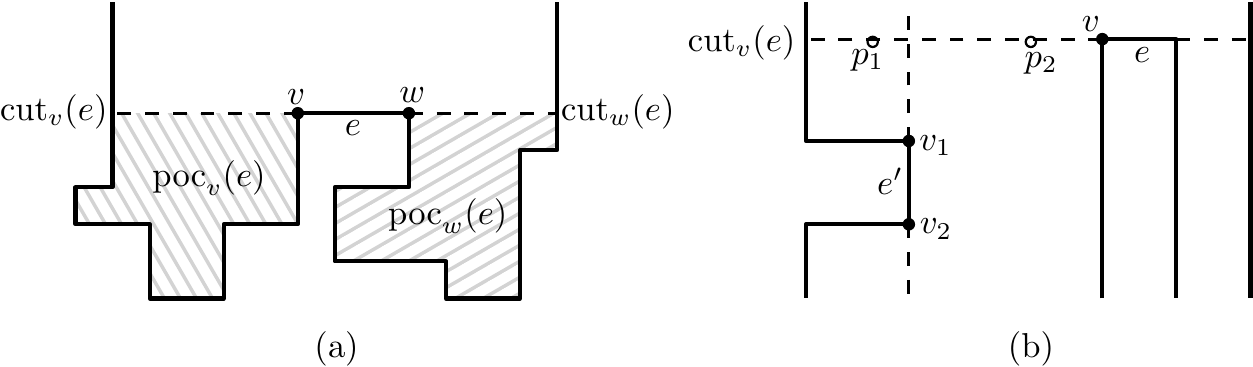}
\caption{(a) The two cuts $\cut_v(e)$ and $\cut_w(e)$ extending a reflex edge $e$
and the pockets of $e$, when $e$ is a top reflex edge.
(b) Proof of Lemma~\ref{lem:monotone_pocket}.
\label{fig:monotone_pocket}}
\end{figure}

\begin{lemma} \label{lem:monotone_pocket}
 Suppose that $P$ is not $xy$-monotone.
 Then there exist a reflex edge $e$ of $P$ and an endpoint $v$ of $e$ such that
 the pocket $\poc_v(e)$ of $e$ at $v$ is $xy$-monotone. 
\end{lemma}
\begin{proof} 
Let $e$ be a reflex edge of $P$, and $v$ be an endpoint of $e$, such that
the number of vertices of the pocket $\poc_v(e)$ is minimum.
If the pocket $\poc_v(e)$ has no reflex edge, then
$\poc_v(e)$ is $xy$-monotone by Observation~\ref{obs:monotone}.

Suppose that $\poc_v(e)$ contains at least one reflex edge $e'$.
We claim that one of the two pockets of $e'$ is always contained in $\poc_v(e)$.
If our claim is true, then such a pocket of $e'$ has fewer vertices
than $\poc_v(e)$ does, a contradiction.

Let $v_1$ and $v_2$ be the two endpoints of $e'$. (See \figurename~\ref{fig:monotone_pocket}b.)
Suppose that our claim is false,
that is, neither $\poc_{v_1}(e')$ nor $\poc_{v_2}(e')$  is contained in $\poc_v(e)$.
It means that each pocket contains points on both sides of  $\cut_v(e)$.
Pockets are simple polygons, and hence they are connected. Thus, $\cut_v(e)$ contains
a point $p_1$ of $\poc_{v_1}(e')$ and a point $p_2$ of $\poc_{v_2}(e')$. So the boundaries
of  $\poc_{v_1}(e')$ and $\poc_{v_2}(e')$ must cross $\cut_v(e)$ between $p_1$ and $p_2$,
which implies that $\cut_{v_1}(e')$ and $\cut_{v_2}(e')$ cross $\cut_v(e)$. This
is impossible, as  $\cut_{v_1}(e')$ and $\cut_{v_2}(e')$ are
collinear.
\end{proof}

We are now ready to prove our upper bound $\lfloor \frac{3r}{4}\rfloor (=\lfloor \frac{3n-4}{8}\rfloor-1)$
on the beacon-based routing problem.
Our proof is by induction on the number $r = r(P)$ of reflex vertices of $P$
based on partitioning $P$ into $xy$-monotone subpolygons in a recursive manner.

\begin{lemma}
\label{lem:routing}
 For any simple rectilinear polygon $P$ with $n$ vertices, $r$ of which are reflex,
 $\lfloor \frac{3n-4}{8}\rfloor-1=\lfloor \frac{3r}{4}\rfloor$ beacons are always sufficient
 to route between all pairs of points in $P$. 
\end{lemma}
\begin{proof}
Our proof is by induction on $r = r(P)$, the number of reflex vertices of $P$.
First consider the base case where $r \leq 1$.
In this case, $P$ is $xy$-monotone by Observation~\ref{obs:monotone},
so no beacon is required for $P$ as discussed above.
Hence, the upper bound holds.

Now, suppose that $P$ is not $xy$-monotone, and thus $r \geq 2$.
Then, Lemma~\ref{lem:monotone_pocket} implies the existence of
a pocket $\poc_v(e)$ of $P$ that is $xy$-monotone, where $v$ is an endpoint of a reflex edge $e$ of $P$.
Without loss of generality, we assume that $e$ is a top reflex edge and $v$ is the left endpoint of $e$.
Let $v'$ be the other endpoint of $e$, so
$\poc_v(e) = P^-_{\cut_v(e)}$ and $\poc_{v'}(e) = P^-_{\cut_{v'}(e)}$.
We consider three subpolygons, as in Figure~\ref{fig:routing}a,
$A := \poc_v(e)$, $C := \poc_{v'}(e)$, and $B := P \setminus (A \cup C)$. Let $p$ be the endpoint other than $v$ of $A\cap B$ and let $q$ be the end point than $v'$ of $C\cap B$. Note that $r(A) + r(C) + r(B) = r - 2$.
We split into two cases: either $r(A) = 0$, or $r(A)\geq 1$.

\begin{figure}[tb]
\centering
\includegraphics[width=0.8\textwidth]{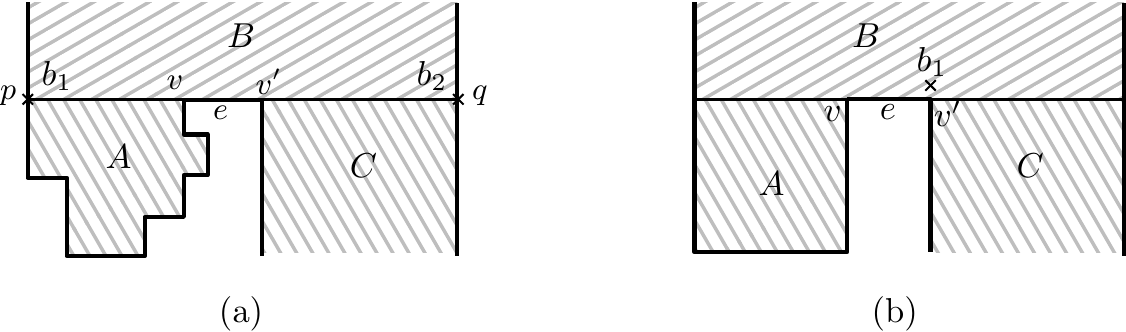}
\caption{Proof of Theorem~\ref{thm:routing}. (a) When $r(A) \geq 1$,
two beacons $b_1$ and $b_2$ are placed at $p$ and $q$ (marked by $\times$).
(b) When $r(A) = 0$, a beacon $b_1$ is first placed just above $v'$ (marked by $\times$), and other beacons will be placed more according to the shape of $C$.}
\label{fig:routing}
\end{figure}

We first consider the case that $r(A) \geq 1$. We place two beacons $b_1$ and $b_2$ at $p$ and $q$, respectively, then place $\lfloor \frac{3r(B)}{4}\rfloor$ beacons in $B$ and $\lfloor \frac{3r(C)}{4}\rfloor$ beacons in $C$ recursively. Since $r(A)\geq 1$, $r(B)+r(C)=r-r(A)-2\leq r-3$. Using this fact, we can bound the total number of beacons we have placed below by
\[ \left\lfloor \frac{3r(B)}{4}\right\rfloor + \left\lfloor \frac{3r(C)}{4}\right\rfloor + 2 \leq \left\lfloor \frac{3(r(B)+r(C))+8}{4} \right\rfloor \leq \left\lfloor \frac{3(r-3)+8}{4} \right\rfloor \leq \left\lfloor \frac{3r}{4} \right\rfloor.\]

We now check if any pair of $s$ and $t$ can be routed via these beacons. 
Segments $vp$, $v'q$, and $pq$ used for the partition are all the convex edges of the corresponding subpolygons, so by Observation~\ref{obs:convex_edge}, none of the segments are hit by any beacon-based routing path between two points in a subpolygon. This implies that once $s$ is routed to some point in (or on the boundary of) a subpolygon, it can be routed to any target $t$ in the subpolygon by the induction hypothesis. $A\cap B$ contains $b_1$ and $B\cap C$ contains $b_2$. Moreover, $b_1$ and $b_2$ can attract each other along $pq$, so any pair of $(s, t)$ for $s, t \in P$ can be routed via $b_1$ or via $b_2$ or both of them. This completes the case where $r(A)\geq 1$.

Now, we consider the other case where $r(A)=0$, which means $A$ is a rectangle. We place a beacon $b_1$ infinitesimally above $v'$ as in Fig.~\ref{fig:routing}b. Note that $b_1$ is placed inside $B$, but it can attract any point in $A$, and it can be attracted to any point in $A$ because it is located above the line connecting $v'$ and the lower right corner of $A$. We place $\lfloor \frac{3r(B)}{4}\rfloor$ beacons in $B$ recursively. For $C$, we need a more careful placement method as follows. 

We first suppose that $r(C) = 0$. Then no beacons inside $C$ are required because $b_1$ can attract any point in $C$ and it can be attracted to any point in $C$. Using this fact, we can easily verify that any two points in $A\cup C$ can be routed to each other via $b_1$. The route between $b_1$ and a point in $B$ is always possible by the induction hypothesis, which implies that any pair $(s, t)$ can be routed wherever $s$ and $t$ belong to. The number of beacons we have placed is at most
$$1+\left\lfloor \frac{3r(B)}{4}\right\rfloor \leq \left\lfloor \frac{3r(B)+4}{4}\right\rfloor \leq \left\lfloor\frac{3(r-2)+4}{4}\right\rfloor \leq \left\lfloor \frac{3r}{4}\right\rfloor,$$
since $r = r(B)+2$. Thus, from now on, we assume that $r(C) > 0$. 

For $r(C) > 0$, we partition $C$ into at most three smaller subpolygons as follows. We sweep the interior of $C$ with an initial sweeping line segment $\ell:= \cut_{v'}(e)$ downward as long as the swept region remains $xy$-monotone. See Fig.~\ref{fig:routing_C}. If the $xy$-monotonicity is violated, then there must be a reflex edge $e'$ that $\ell$ intersects. Then $e'$ would be either horizontal or vertical. 

\begin{figure}[tb]
\centering
\includegraphics[width=\textwidth]{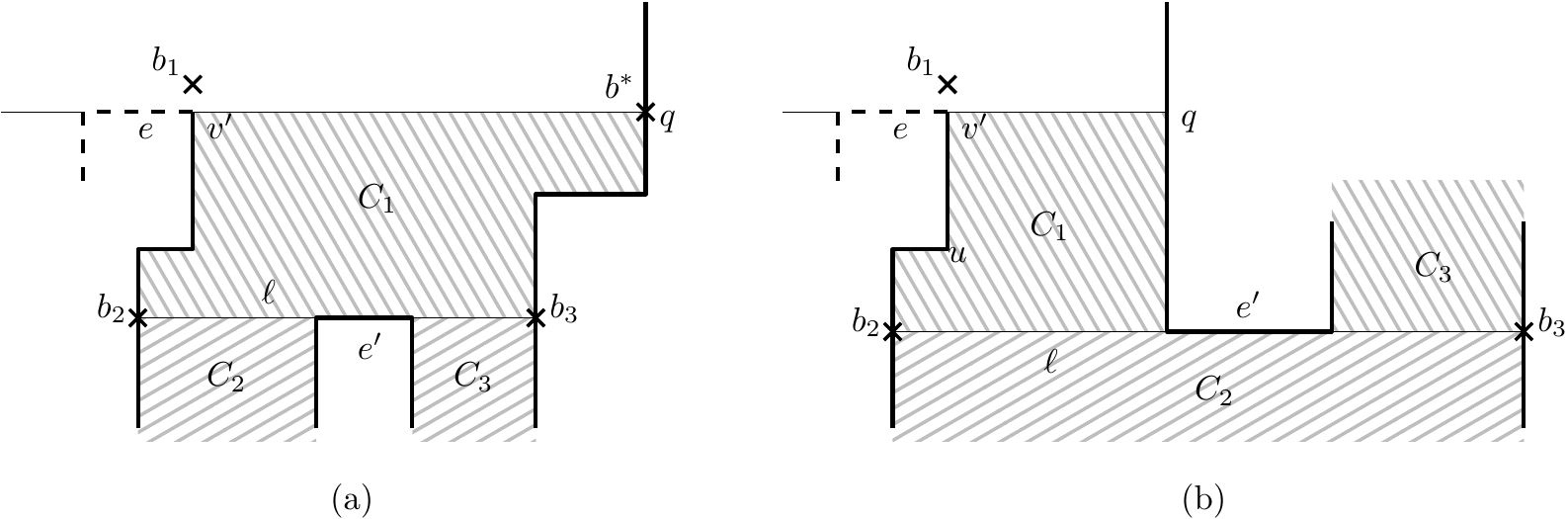}
\caption{The case where the reflex edge $e'$ is horizontal. (a) $e'$ is a bottom reflex edge. The beacon $b^*$ is placed because $r(C_1)\geq 2$. (b) $e'$ is a top reflex edge, which is connected from $q$. The beacon $b^*$ is not placed because $r(C) = 1 < 2$. The symmetric case where a top reflex edge $e'$ is connected from $v'$ is omitted in this figure.}
\label{fig:routing_C}
\end{figure}

Let us suppose that $e'$ is horizontal. For this case, $e'$ could be either a bottom reflex edge as in Fig.~\ref{fig:routing_C}a or a top reflex edge as in Fig.~\ref{fig:routing_C}b. Then we split $C$ into three subpolygons by cutting $C$ along the sweeping segment $\ell$ containing $e'$; $C_1$ is the $xy$-monotone piece, $C_2$ and $C_3$ are the other two pieces as in Fig.~\ref{fig:routing_C}a-b. We place the beacons in $C_2$ and $C_3$ (not in $C_1$) recursively, and place additional beacons as follows. We place two beacons $b_2$ and $b_3$ at two end points of $\ell$, where $b_2$ is assumed to be in the left of $b_3$. We place one more beacon $b^*$ at the point $q$ only when $r(C_1)\geq 2$. It holds that $r = r(B)+r(C_1)+r(C_2)+r(C_3)+4$. When $r(C_1)\geq 2$, the total number of beacons we have placed is
\[ 4+\left\lfloor \frac{3(r(B)+r(C_2)+r(C_3))}{4}\right\rfloor \leq \left\lfloor \frac{3(r-r(C_1)-4)+16}{4}\right\rfloor \leq \left\lfloor \frac{3r-2}{4} \right\rfloor \leq \left\lfloor \frac{3r}{4} \right\rfloor.\]
When $r(C_1)<2$, the beacon $b^*$ is not placed, so the number of the beacons is
\[ 3+\left\lfloor \frac{3(r(B)+r(C_2)+r(C_3))}{4}\right\rfloor \leq \left\lfloor \frac{3(r-4)+12}{4}\right\rfloor \leq \left\lfloor \frac{3r}{4} \right\rfloor.\]

Let us check if $s$ can be routed to $t$ in this placement. As in the previous case where $r(A)\geq 1$, all the segments used for the partition are convex edges in their associated subpolygons. So, it is sufficient to show that any pair of beacons from $\{b_1, b_2, b_3, b^*\}$ can be routed to each other. First, $b_2$ and $b_3$ are attracted to each other along the cut $\ell$, and so are $b_1$ and $b^*$ because they are visible each other. Second, $b_2$ or $b_3$ is on the boundary of $C_1$, and $C_1$ is $xy$-monotone, so if $r(C_1)\geq 2$, i.e., $b^*$ exists, then $b^*$ can be routed to $b_2$ or $b_3$. Otherwise, if $r(C_1)<2$ as in Fig.~\ref{fig:routing_C}b, then $C_1$ is either a rectangle or a union of two rectangles with the unique reflex vertex $u$. We here claim that $b_1$ can be attracted to $b_2$ or $b_3$, and $b_1$ can also attract $b_2$ or $b_3$. If $b_1$ can see directly the one of them, then it is done. Suppose that they are not visible from $b_1$. For this to happen, $u$ and an end vertex of $e'$ must obstruct the sight from $b_1$ to $b_2$ and to $b_3$. Then there are four situations as in Fig.~\ref{fig:routing_C1}. In the first three situations (Fig.~\ref{fig:routing_C1}a-c), $b_1$ can be clearly attracted to $b_2$ or $b_3$. For the last situation as in Fig.~\ref{fig:routing_C1}d, it cannot be attracted to $b_2$, but it can be attracted to $b_3$ via the end vertex of $e'$ and along the cut containing $e'$. Thus $b_1$ can always reach $b_2$ or $b_3$. The reverse attraction is also possible. If $b_1$ can see $b_2$ or $b_3$, say $b_2$, then $b_3$ first goes to $b_2$ then reaches $b_1$. Otherwise, only the last situation in Fig.~\ref{fig:routing_C1}d would be in trouble because $b_3$ cannot be attracted to $b_1$. But $b_2$ can be attracted to $b_1$, thus $b_3$ can reach $b_1$ via $b_2$. As a result, the beacons in $\{b_1, b_2, b_3, b^*\}$ can be routed to each other, which means that any pair $(s, t)$ can be routed in this partition.

\begin{figure}[tb]
\centering
\includegraphics[width=\textwidth]{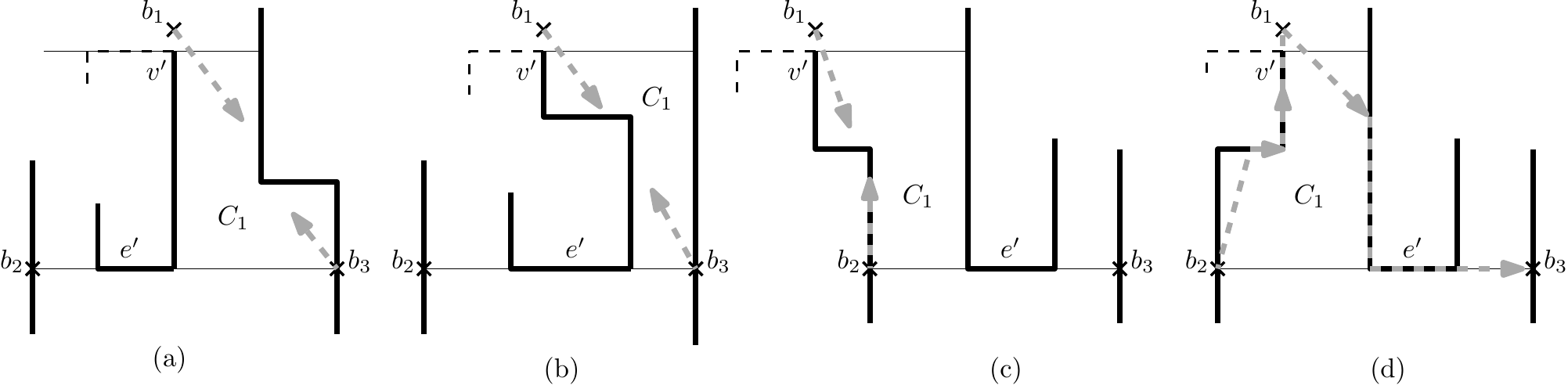}
\caption{Four different situations that $b_1$ sees neither $b_2$ nor $b_3$.}
\label{fig:routing_C1}
\end{figure}

We now consider the last case where $e'$ is a vertical reflex edge. See Fig.~\ref{fig:routing_last}. Let $w$ be the lower end vertex of $e'$, and let $e''$ be the horizontal edge incident to $w$ with the other end vertex $w'$. If $e''$ is a reflex edge, that is, $w'$ is reflex, then it should be a top reflex edge, thus we cut $C$ along $\ell$ into three pieces $C_1$, $C_2$, and $C_3$, where $C_3$ is a pocket $\poc_{w'}(e'')$, and $C_2 := C\setminus (C_1\cup C_3)$. We place two beacons $b_2$ and $b_3$ at two end points of $\ell$. We place $b^*$ at $q$ only when $r(C_1)\geq 2$. Finally, we place the beacons in $C_2$ and $C_3$ recursively. If $e''$ is not a reflex edge, i.e., $w'$ is convex, then we cut $C$ along $\cut_{w}(e'')$ into two pieces $C_1$ and $C_2$, where $C_2:= C\setminus C_1$. We place a beacon $b_2$ at the endpoint of $\cut_{w}(e'')$ (not at $w$), and one more beacon $b^*$ at $q$ regardless of the size of $r(C_1)$. Note that $b_2$ and $b^*$ are located both at the convex vertices of $C_1$.

\begin{figure}[tb]
\centering
\includegraphics[width=0.8\textwidth]{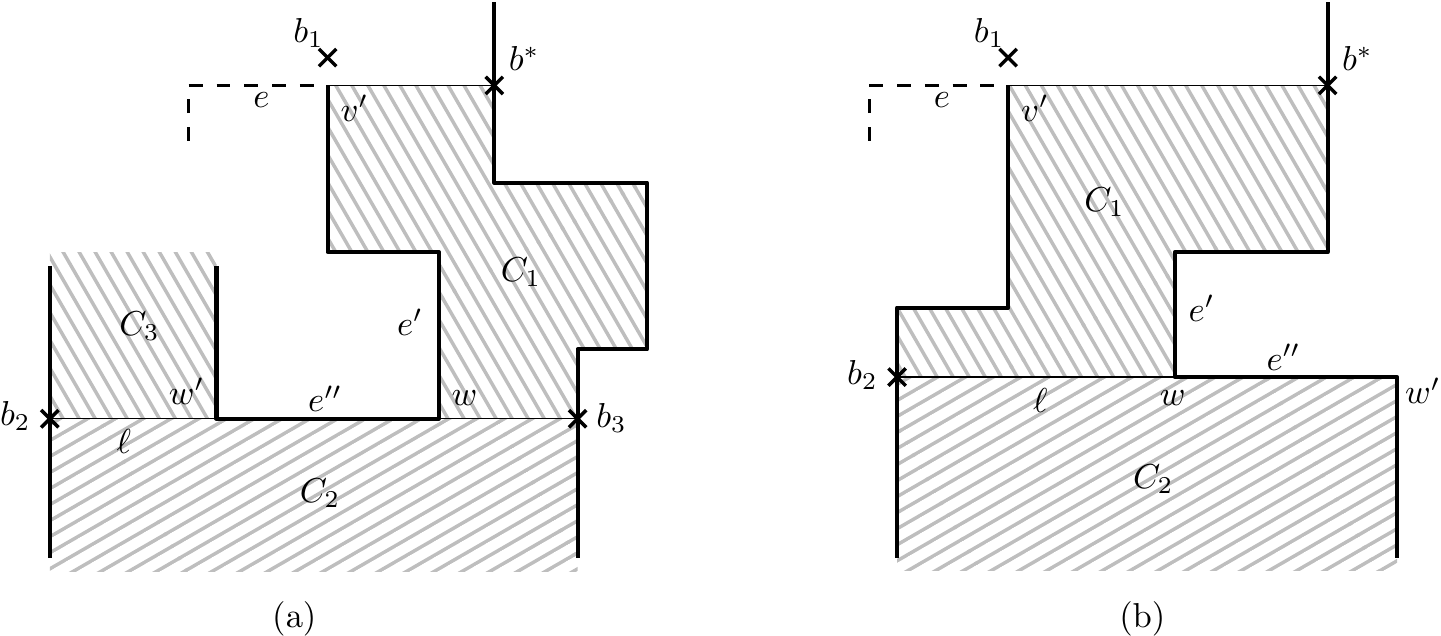}
\caption{The case where $e'$ is a vertical reflex edge. The symmetric cases are omitted in this figure.}
\label{fig:routing_last}
\end{figure}

The partition $(C_1, C_2, C_3)$ actually corresponds to the previous case where $e'$ is horizontal, so we can apply the same arguments to bound the number of beacons and to ensure that any pair $(s, t)$ can be routed. We now focus only on the partition $(C_1, C_2)$. We have placed three beacons $b_1, b_2, b^*$. Note here that $r(C_1)\geq 1$ because $C_1$ always contains the upper reflex vertex of $e'$. Using this with the fact that $r = r(B)+r(C_1)+r(C_2)+3$, we can bound the number of beacons by
\[
3+\left\lfloor \frac{3(r(B)+r(C_2))}{4}\right\rfloor \leq \left\lfloor \frac{3(r-r(C_1)-3)+12}{4}\right\rfloor \leq \left\lfloor \frac{3(r-4)+12}{4} \right\rfloor \leq \left\lfloor \frac{3r}{4} \right\rfloor.
\] Let us check if $s$ can be routed to $t$. The cuts used for the partition are again served as convex edges in associated subpolygons. Thus it suffices to prove that three beacons $b_1, b_2, b^*$ are routed to each other. The two beacons $b_1$ and $b^*$ are visible, so they can attract each other. Because $b^*$ and $b_2$ are both in the $xy$-monotone subpolygon $C_1$, they also attract each other. Thus the three beacons can attract each other. This completes the proof of the theorem.
\end{proof}

%

%
%
\section{Concluding Remarks} \label{sec:conclusion}

In this paper, we attempt to reduce the gaps between the lower and upper bounds on the number of beacons required in beacon-based coverage and routing problems for a simple rectilinear polygon $P$. For the coverage problem, we raised its lower bound, and presented an algorithm to place the same number of beacons to cover $P$. These results settle the open questions on the coverage problem. For the routing problem, we improved the lower and upper bounds, but there is still a gap between them, which is an immediate open question. Furthermore, we presented an optimal linear time algorithm of computing the beacon-based kernel of $P$. But it remains open to compute the \emph{inverse kernel} of $P$, which is defined as a set of points in $P$ that are attracted to all the points in $P$, in a subquadratic time.


\begin{thebibliography}{10}

\bibitem{b-bbrg-13}
M.~Biro.
\newblock {\em Beacon-based routing and guarding}.
\newblock Dissertation, Stony Brook University, 2013.

\bibitem{bgikm-cccg-13}
M.~Biro, J.~Gao, J.~Iwerks, I.~Kostitsyna, and J.~S.~B. Mitchell.
\newblock Combinatorics of beacon-based routing and coverage.
\newblock {\em Proc. the 25th Canadian Conf. Comput. Geom. (CCCG 2013)}, 2013.

\bibitem{bikm-wads-13}
M.~Biro, J.~Iwerks, I.~Kostitsyna, and J.~S.~B. Mitchell.
\newblock Beacon-based algorithms for geometric routing.
\newblock In {\em Proc. the 13th WADS (WADS 2013)}, volume 8037 of {\em LNCS},
  pages 158--169, 2013.

\bibitem{c-ctpg-75}
V.~Chav\'{a}tal.
\newblock A combinatorial theorem in plane geometry.
\newblock {\em J. Combinat. Theory Series B}, 18:39--41, 1975.

\bibitem{g-spragt-86}
E.~Gy\"{o}ri.
\newblock A short proof of the rectilinear art gallery theorem.
\newblock {\em SIAM J. on Algebraic and Discrete Methods}, 7(3), 1986.

\bibitem{ghks-ggprp-96}
E.~Gy\"{o}ri, F.~Hoffmann, K.~Kriegel, and T.~Shermer.
\newblock Generalized guarding and partitioning for rectilinear polygons.
\newblock {\em Comput. Geom.: Theory Appl.}, 6(1):21--44, 1996.

\bibitem{kkk-tgrfw-83}
J.~Kahn, M.~Klawe, and D.~Kleitman.
\newblock Traditional galleries require fewer watchmen.
\newblock {\em SIAM J. on Algebraic and Discrete Methods}, 4(2), 1983.

\bibitem{krs-cccg-14}
B.~Kouhestani, D.~Rappaport, and K.~Salmoaa.
\newblock Routing in a polygonal terrain with the shortest beacon watchtower.
\newblock {\em Proc. the 26th Canadian Conf. Comput. Geom. (CCCG 2014)}, 2014.

\bibitem{mp-agtgg-03}
T.~S. Michael and V.~Pinciu.
\newblock Art gallery theorems for guarded guards.
\newblock {\em Comput. Geom.: Theory Appl.}, 26:247--258, 2003.

\bibitem{o-apragt-83}
J.~O'Rourke.
\newblock An alternative proof of the rectilinear art gallery theorem.
\newblock {\em J. Geometry}, 21:118--130, 1983.

\bibitem{o-agta-87}
J.~O'Rourke.
\newblock {\em Art Gallery Theorems and Algorithms}.
\newblock International Series of Monographs on Computer Sciences. Oxford
  University Press, 1987.

\bibitem{s-rrag-92}
T.~Shermer.
\newblock Recent results in art galleries.
\newblock {\em IEEE Proceedings}, 90(9), 1992.

\bibitem{u-agip-00}
J.~Urrutia.
\newblock Art gallery and illumination problems.
\newblock In J.-R. Sack and J.~Urrutia, editors, {\em Handbook of Computational
  Geometry}, chapter~22, pages 973--1027. North-Holland, 2000.

\end{thebibliography}

\end{document}